\documentclass[]{interact}

\usepackage{epstopdf}
\usepackage[caption=false]{subfig}

\usepackage[numbers,sort&compress]{natbib}
\bibpunct[, ]{[}{]}{,}{n}{,}{,}
\makeatletter
\def\NAT@def@citea{\def\@citea{\NAT@separator}}
\makeatother

\theoremstyle{plain}
\newtheorem{theorem}{Theorem}[section]

\newtheorem{corollary}[theorem]{Corollary}

\theoremstyle{definition}
\newtheorem{definition}[theorem]{Definition}
\newtheorem{example}[theorem]{Example}

\theoremstyle{remark}
\newtheorem{remark}{Remark}

\usepackage{amsmath}
\usepackage{nicefrac} 
\usepackage{multirow} 
\usepackage{mathtools}

\usepackage{caption}
\usepackage{booktabs}
\usepackage{threeparttable}
\usepackage{bbm}
\usepackage{xcolor}
\usepackage{titlesec}

\DeclareMathOperator{\sign}{sign}

\DeclarePairedDelimiter\floor{\lfloor}{\rfloor}

\begin{document}


\title{Indirect inference for locally stationary ARMA processes with stable innovations}

\author{
\name{Shu Wei Chou-Chen\thanks{CONTACT Shu Wei Chou-Chen Email: shuchou@alumni.usp.br} and Pedro A. Morettin}
\affil{Institute of Mathematics and Statistics, University of S\~ao Paulo, Brazil}
}

\title{Indirect inference for locally stationary ARMA processes with stable innovations}


\maketitle

\bigskip

\begin{abstract}
The class of locally stationary processes assumes that there is a time-varying spectral representation, that is, the existence of finite second moment. We propose the $\alpha$-stable locally stationary process by modifying the innovations into stable distributions and the indirect inference to estimate this type of model. Due to the infinite variance, some of interesting properties such as time-varying autocorrelation cannot be defined. However, since the $\alpha$-stable family of distributions is closed under linear combination which includes the possibility of handling asymmetry and thicker tails, the proposed model has the same tail behavior throughout the time. In this paper, we propose this new model, present theoretical properties of the process and carry out simulations related to the indirect inference in order to estimate the parametric form of the model. Finally, an empirical application is illustrated.
\end{abstract}


\section{Introduction}

\label{sec:intro}

The class of locally stationary processes describes processes that are approximately stationary in a neighborhood of each time point but its structure, such as covariances and parameters, gradually changes throughout the time period \cite{Dahlhaus1996a,Dahlhaus1997}. This type of processes has been proved to achieve meaningful asymptotic theory by applying infill asymptotics. The idea of this approach is that the time-varying parameters are rescaled to the unit interval, and thus, more available observations imply obtaining more contribution for each local structure. Consequently, statistical asymptotic results such as consistency, asymptotic normality, efficiency, locally asymptotically normal expansions, etc. are obtained. \cite{Dahlhaus2012} provided a review of this type of process. 

Most results of locally stationary processes assume innovations with finite second moment. However, different areas have observed phenomena with heavy tail distributions or infinite variance. In this work, we consider that the innovations of the locally stationary process follow $\alpha-$stable distributions. The advantage of assuming $\alpha-$stable distributions is its flexibility for asymmetry and thick tails. Also, it is closed under linear combinations and includes the Gaussian distribution as a special case. However, its estimation is difficult since the density function does not have a closed-form. Consequently, the usual estimation methods such as maximum likelihood and method of moments do not work. 

Alternative estimation approaches such as methods based on quantiles \citep{McCulloch1986} or on the empirical characteristic function \citep{Koutrouvelis1981} are proposed. However, those methods are only useful for the estimation of the $\alpha-$stable distributions parameters and, therefore, they are difficult to apply for more complex models. 

The strategy to estimate this kind of process is the indirect inference proposed by \cite{Gourieroux1993} and \cite{Gallant1996}. Since $\alpha-$stable distributions can be easily simulated, the indirect approach, which is an intensive computationally simulation based method, can be a solution to overcome the estimation problem.

Models involving stable distribution were successfully implemented in indirect inference for independent samples from the $\alpha$-stable distributions and $\alpha-$stable ARMA processes \cite{Lombardi2008}. Moreover, some time series models involving stable distributions are also successfully implemented using indirect inference \citep{Sampaio2015,Sampaio2019,Calzolari2014,Calzolari2018}.

Our contribution in this work is twofold. First, we propose the locally stationary processes with stable innovations and present the theoretical properties of this model. We also justify the reason why we call them $\alpha-$stable locally stationary processes. Second, we propose the indirect inference in order to estimate the models with linear time-varying coefficient.

The paper is organized as follows. In Section \ref{sec:background}, we review the basic background on locally stationary processes, $\alpha-$stable distribution and indirect inference. Then, properties of the $\alpha-$stable locally stationary processes are presented in Section \ref{sec:stable_tvARMA}. Section \ref{sec:indirect_inference_stable_tvARMA} describes the indirect inference for this kind of processes. Simulations are performed to study the indirect inference approach in Section \ref{sec:simulation}. A wind data application is illustrated in Section \ref{sec:application}. Finally, conclusions are presented in Section \ref{sec:conclusion}.

\section{Background}
\label{sec:background}
\subsection{Locally stationary processes}

Locally stationary processes were introduced by using a time-varying spectral representation \cite{Dahlhaus1997}. However, we use the time-domain version as in \cite{Dahlhaus2009} since stable distributions do not have finite second moment.

\begin{definition} 
	\label{def:LLSP1}
	The sequence of stochastic processes $X_{t,T}$  $(t=1,...,T)$ is a linear locally stationary processes if $X_{t,T}$ has a representation
	\begin{equation}
	\label{LLSP1}
	X_{t,T}= \sum_{j=-\infty}^{\infty} a_{t,T}(j) \varepsilon_{t-j},
	\end{equation}	
	where the following conditions are satisfied:
	
	\begin{enumerate}
		\item [(i)]
		\begin{equation}
		\label{assumption1_i}
		\sup_{t} \left| a_{t,T} (j) \right| \leq \frac{K}{\ell (j)},~~~\text{with K independent of T;}
		\end{equation}
		
		\item [(ii)]
		\label{assumption1_ii}Let $V(g)$ be the total variation of a function $g$ on $\left[0,1 \right] $; then, there exist funcions $\alpha(\cdot,j): \left(0,1\right] \rightarrow \mathbb{R}$ with
		
		\begin{equation}
		\label{assumption1_iia}
		\sup_{u} \left| a(u,j) \right| \leq \frac{K}{\ell (j)},
		\end{equation}				
		\begin{equation}
		\label{assumption1_iib}
		\sup_{j} \sum_{t=1}^{T} \left| a_{t,T}(j) - a\left( \frac{t}{T},j\right)  \right| \leq K,
		\end{equation}		
		\begin{equation}
		\label{assumption1_iic}
		V( a(\cdot,j)) \leq \frac{K}{\ell(j)};
		\end{equation}		
		\item [(iii)]
		\label{assumption1_iii} $\left\lbrace \varepsilon_t\right\rbrace $ are i.i.d. with $E\left[ \varepsilon_t\right] =0$, $E\left[ \varepsilon_s,\varepsilon_t \right]=0$ for $s \neq t$ and $E\left[ \varepsilon_t^2 \right]=1$.
	\end{enumerate}
\end{definition}

Note that the condition (iii) in the Definition \ref{def:LLSP1} assumes that the innovations $\left\lbrace \varepsilon_t \right\rbrace $ has zero mean and unit variance. In this case, there exists a spectral representation of the process $X_t$ and time-varying spectral density. Note that classical stationary processes arise as a special case when all parameter curves are constant. 

\subsection{$\alpha-$stable distribution}

Stable distributions, as an extension of Gaussian distributions, can be defined by its characteristic function:
\begin{equation}
\label{stable_ch.func.}
E \left\lbrace e^{i \theta X} \right\rbrace = \left \{ \begin{matrix} \exp \left\lbrace -\sigma^\alpha |\theta|^\alpha (1-i \beta (\sign \theta) \tan \frac{\pi \alpha}{2})+i\mu\theta \right\rbrace, & \text{if $\alpha \neq 1$,}
\\ \exp \left\lbrace -\sigma |\theta| (1+i \beta \frac{2}{\pi} (\sign \theta) \log |\theta|)+i\mu\theta \right\rbrace, & \text{if $\alpha = 1$,}\end{matrix}\right. 
\end{equation}
where $\alpha \in \left( 0,2\right]$ is the index of stability (tail heaviness), $\sigma > 0$ is the scale parameter, $-1 \leq \beta \leq 1$ the asymmetry parameter and $\mu \in \mathbb{R}$ the location parameter. It is denoted by $S_\alpha (\sigma,\beta,\mu)$. Important properties can be found in detail in \cite{Samorodnitsky1994}. 

This class of distribution generalizes some important known distributions: normal distribution for $\alpha=2$, Cauchy distribution ($\alpha=1,~\beta=0$) and L\'evy distribution ($\alpha=1/2,~\beta=\pm 1$). However, it does not have a closed-form density function in general. Moreover, the non-existence of moments greater than $\alpha$ makes it difficult to estimate parameters. Different estimation approaches have been proposed, such as methods based on quantiles \citep{McCulloch1986} and methods based on the empirical characteristic function \citep{Koutrouvelis1981}. Nevertheless, they are only useful for independent samples from stable distributions and difficult to perform for more complex models.

When a random variable has the density function and distribution function, its simulation is an easy task. In stable distribution case, \cite{Weron1995} proposed an algorithm to generate $\alpha$-stable distribution. To simulate a random variable $X \sim S_\alpha (1,\beta,0)$:
\begin{enumerate}
	\item generate a random variable $U$ uniformly distributed on $\left( -\frac{\pi}{2},\frac{\pi}{2} \right) $, and a independent exponential random variable $W$ with mean $1$, then
	\item let $B_{\alpha,\beta}=\frac{\arctan\left(\beta \tan \frac{\pi \alpha}{2}\right)}{\alpha}$ and $S_{\alpha, \beta}= \left[ 1+ \beta^2 \tan^2 \frac{\pi \alpha}{2} \right]^{1/(2\alpha)} $ , and compute
	\begin{equation}
	X=
	\left \{ \begin{matrix} S_{\alpha,\beta} \frac{\sin \left( \alpha (U+B_{\alpha,\beta}) \right) }{(\cos U)^{1/\alpha}} \left[ \frac{\cos(U-\alpha(U+B_{\alpha,\beta}))}{W} \right] ^{\frac{1-\alpha}{\alpha}},  & \text{if $\alpha \neq 1$,}
	\\ \frac{2}{\pi} \left[ \left( \frac{\pi}{2}+\beta U \right) \tan U - \beta \log \left( \frac{\frac{\pi}{2}W \cos U}{\frac{\pi}{2}+\beta U} \right)   \right],  &\text{if $\alpha = 1$.}
	\end{matrix}\right.
	\end{equation}
\end{enumerate}
Next, $Y \sim S_\alpha (\sigma,\beta,\mu)$ is obtained by means of the standardization formula:
\begin{equation}
Y=\left \{ 
\begin{array}{ll}
\sigma X +\mu,  & \text{if $\alpha \neq 1$,}
\\ \sigma X+ \frac{2}{\pi} \beta \sigma \log \sigma +\mu, ~~ &\text{if $\alpha = 1$.}
\end{array}\right.
\end{equation}

Since stable distributions can be easily simulated, the indirect approaches proposed by \cite{Gourieroux1993} and \cite{Gallant1996} could be the solution to more complex models involving stable distributions.

\subsection{Indirect inference}

The indirect inference proposed by \cite{Gourieroux1993} is based on a very simple idea and it is suitable for situations where the direct estimation of the model is difficult. Let $\boldsymbol{y}$ be a sample of $T$ observations from a model of interest (IM) and $\hat{\theta}$ be the maximum likelihood estimator (MLE) of $\theta \in \Theta$ in IM which is unavailable. Then, consider the auxiliary model (AM) depending on a parameter vector $\lambda \in \Lambda$ whose likelihood function is easier to handle, but its MLE $\hat{\lambda}$ is not necessarily consistent. The indirect inference is carried out by the following steps:

\begin{description}
	\item [Step 1] Compute $\hat{\lambda}$ based on $\boldsymbol{y}$.
	\item [Step 2] Simulate a set of $S$ vectors of size $T$ from the IM on the basis of an arbitrary parameter vector $\hat{\theta}^{(0)}$. Let us denote each of those vectors as $y^s(\hat{\theta}^{(0)}),~s=1,...,S$.
	\item [Step 3] Then, estimate parameters of the AM using simulated values from the IM,
	\begin{equation}
	\hat{\lambda}_S(\hat{\theta}^{(0)}) = \operatorname*{argmax}_{\lambda \in Z} \sum_{s=1}^{S} \sum_{t=1}^{T} \ln \tilde{\mathcal{L}} \left[ \lambda;y^s(\hat{\theta}^{(0)})\right].
	\end{equation}
	\item [Step 4] Numerically update the initial guess $\hat{\theta}^{(0)}$ in order to minimize the distance
	\begin{equation}
	\left[ \hat{\lambda}-\hat{\lambda}_S \right]' \Omega \left[ \hat{\lambda}-\hat{\lambda}_S \right],
	\end{equation}
	where $\Omega$ is a symmetric nonnegative matrix defining the metric.	
\end{description}

For choosing $\Omega$, \cite{Gourieroux1993} proved that when the parameter vectors of both AM and IM have the same dimension and $T$ is sufficiently large, the estimator does not depend on the  matrix $\Omega$. In this paper, we consider $\Omega$ as identity matrix. Finally, the estimation step is performed with a numerical algorithm, such as Newton-Raphson. Then, for a given estimate $\hat{\theta}^{(p)}$, the procedure yields $\hat{\theta}^{(p+1)}$ and the process will be repeated until the series of $\hat{\theta}^{(p)}$ converges. The estimator is then given by
\begin{equation}
	\hat{\theta} = \lim\limits_{p \rightarrow \infty} \hat{\theta}^{(p)}.
\end{equation}

\section{tvARMA with stable innovations}
\label{sec:stable_tvARMA}
An important example of the locally stationary process is the time-varying ARMA model, briefly tvARMA. In this section, we will consider this model with stable innovations.

\begin{definition}[tvARMA with stable innovations]
	\label{def:stable_tvARMA}
	Consider the system of difference equations
	\begin{equation}
	\label{eq:tvARMA_stable}
	\sum_{j=0}^{p} \alpha_j\left( \frac{t}{T} \right)  X_{t-j,T}=\sum_{k=0}^{q} \beta_k\left( \frac{t}{T} \right)\gamma\left( \frac{t-k}{T} \right)  \varepsilon_{t-k},
	\end{equation}	
	where $\varepsilon_t$ are i.i.d. and $\varepsilon_t \sim S_\alpha (\nicefrac{1}{\sqrt{2}},\beta,0)$ with $\alpha \in (0,2)$. Assume $\alpha_0(u) \equiv \beta_0(u) \equiv 1$ and $\alpha_j(u)=\alpha_j(0)$, $\beta_k(u)=\beta_k(0)$ for $u<0$. Suppose also that all $\alpha_j(\cdot)$ and $\beta_k(\cdot)$, as well as $\gamma^2(\cdot)$, are of bounded variation. 
\end{definition}

The reason that the scale parameter of the innovations is set to be $\sigma = \nicefrac{1}{\sqrt{2}}$ is when $\alpha=2$, the standardized Gaussian innovation is obtained. It is possible to define the equation \eqref{eq:tvARMA_stable} as:
\begin{equation}
\label{eq:tvARMA_stable_matrix}
\Phi_{t,T}(B)X_{t,T}=\Theta_{t,T}(B) z_{t,T},
\end{equation}
where $z_{t,T}=\gamma (\frac{t}{T}) \varepsilon_{t}$; $\Phi_{t,T}(B)=1+\alpha_1( \frac{t}{T} )B +\cdots +\alpha_p( \frac{t}{T})B^p $ and $\Theta_{t,T}(B)=1+\beta_1( \frac{t}{T} )B +\cdots +\beta_q( \frac{t}{T})B^q $ are the autoregressive (AR) and moving average (MA) operators, respectively. 

There are several works related to stable linear processes. For instance, chapter 7 in \cite{Embrechts1997} and Chapter 13 in \cite{Brockwell1991} give a general review of stable linear processes. \cite{Kokoszka1994} study the infinite variance stable ARMA procseses and \cite{Kokoszka1995} study fractional ARIMA with stable innovations. \cite{Mikosch1995} proposed a Whittle-type estimator to estimate the coefficients of the ARMA model. In the stable innovation and time-varying coefficient context, \cite{Peiris2001a,Peiris2001b} considered the univariate and multivariate case of the system \eqref{eq:tvARMA_stable_matrix} with symmetric stable innovations and assume $\gamma(\cdot)=1$. However, they considered time-dependent coefficient without the local stationarity condition.

\subsection{Existence and Uniqueness of a Solution}

Before we study the local stationarity conditions on the time-varying coefficients, we present a set of regularity conditions of existence and uniqueness of solution of the system based on the concepts defined by \cite{Peiris2001a,Peiris2001b}.
\begin{definition}
	\label{def:ARMAregular}
	\begin{enumerate}
		\item The process \eqref{eq:tvARMA_stable_matrix} is AR regular (or causal) if there exist $a_{t,T}(j)$ such that
		\begin{equation}
		\label{eq:stable_ma_inf}
		X_{t,T}= \sum_{j=0}^{\infty} a_{t,T}(j) \varepsilon_{t-j}.
		\end{equation}
		satisfying $\sum\limits_{j=0}^{\infty} |a_{t,T}(j)|^{\delta}<\infty$ for all $t$ and $\delta=min\{1,\alpha\}$.
		\item The process \eqref{eq:tvARMA_stable_matrix} is MA regular (or invertible) if there exist $b_{t,T}(j)$ such that
		\begin{equation}
		\label{eq:stable_ar_inf}
		\varepsilon_{t}= \sum_{j=0}^{\infty} b_{t,T}(j) X_{t-j,T}.
		\end{equation}
		satisfying $\sum\limits_{j=0}^{\infty} |b_{t,T}(j)|^{\delta}<\infty$ for all $t$ and $\delta=min\{1,\alpha\}$.
	\end{enumerate}
\end{definition} 

The random series in \eqref{eq:stable_ma_inf} converges a.s. if and only if ${\sum\limits_{j=0}^{\infty} \left|  a_{t,T}(j)\right| ^\alpha <\infty}$, and by applying the Proposition 13.3.1 in \cite{Brockwell1991}, it converges absolutely if and only if $ \sum\limits_{j=0}^{\infty} \left|  a_{t,T}(j)\right| ^\delta <\infty$ with $\delta=\min \{ 1,\alpha \}$. Similar arguments are applied to the MA representation in \eqref{eq:stable_ar_inf}.

To continue, we omit the subscript $T$ from above notation. Consider the homogeneous difference equation
\begin{equation}
\label{eq:homo_dif_equation}
\Phi_{t}(B) u_t=0.
\end{equation}
If $\alpha_p(\frac{t}{T}) \neq 0$ for any $t$, there exist $p$ linearly independent solution $\psi_{1,t},\psi_{2,t},...,\psi_{p,t}$ such that
\begin{equation}
\label{eq:solution_homo_dif_equation}
\Psi(t)=
\begin{bmatrix} 
\psi_{1,t} & \cdots & \cdots &  \psi_{p,t} \\
\psi_{1,t-1} & \ddots & & \psi_{p,t-1} \\
\vdots &  &   \ddots & \vdots \\
\psi_{1,t-p+1} & \cdots & \cdots & \psi_{p,t-p+1}
\end{bmatrix}
\end{equation}
is invertible for any $t$ \cite{Miller1968}. Therefore, we can define 
\begin{equation}
\label{eq:green_matrix}
G(t,s)=\Psi(t)\left[ \Psi(s)\right]',
\end{equation}
the one-sided Green's function matrix associated with the AR operator $\Phi_{t}(B)$. It can be showed that $G(t,s)$ is unique and invariant under different solutions $\Psi(t)$ obtained from the homogeneous difference equation \eqref{eq:homo_dif_equation}. Furthermore, the one-sided Green's function associated with the AR operator $\Phi_{t}(B)$ is defined as the upper left-hand element in the matrix \eqref{eq:green_matrix},
\begin{equation}
\label{eq:green_function}
g(t,s)=\left[ G(t,s) \right]_{11} ,
\end{equation}
which is also unique and invariant. Now, we are ready to establish the conditions for AR regularity and MA regularity.

\begin{theorem}
	Let $\left\lbrace X_{t,T} \right\rbrace $ be a sequence of stochastic process that satisfies \eqref{eq:tvARMA_stable_matrix}. Suppose that $\alpha_p(\frac{t}{T}) \neq 0$ for all $t$, and $g(t,s)$, the one-sided Green's functions associated with $\Phi_{t}(B)$, is such that $\sum\limits_{s=-\infty}^{t} |g(t,s)|^\delta <\infty$, for all $t$. Assume also that $\sum\limits_{s=-0}^{q} |\beta_j(\cdot)|^2 < \infty$ for all $t$, and $\Phi_{t}(z)$ $(\Phi_{t}(z) \neq 0$ for $|z|\leq 1)$ and $\Theta_{t}(z)$ have no common roots. Then, there is a valid solution, given by
	\begin{equation}
	X_{t,T}= \sum_{j=0}^{\infty} a_{t,T}(j) \varepsilon_{t-j},
	\end{equation}
	to \eqref{eq:tvARMA_stable_matrix} with coefficients uniquely determined by 
	\begin{equation}
	a_{t,T}(j)=
	\left\{ \begin{array}{l r} 
	0, & j<0,\\
	\gamma(\frac{t-j}{T}), & j=0, \\
	\gamma(\frac{t-j}{T}) \sum\limits_{j=0}^k \beta_k (\frac{t-j+k}{T}) g(t,t-j+k), & 0\leq j \leq q, \\
	\gamma(\frac{t-j}{T}) \sum\limits_{j=0}^q \beta_k (\frac{t-j+k}{T}) g(t,t-j+k),  & j>q. 
	\end{array}	
	\right.
	\end{equation}
\end{theorem}

\begin{proof}
	By setting $z_{t,T}=\gamma(\frac{t}{T})\varepsilon_{t}$, along with the absolute convergence conditions above, the proof is similar to \cite{Peiris2001b}. 
\end{proof}

\begin{theorem}
	Let $\left\lbrace X_{t,T} \right\rbrace $ be a sequence of stochastic process that satisfies \eqref{eq:tvARMA_stable_matrix}. Suppose that $\beta_q(\frac{t}{T}) \neq 0$ for all $t$, and $h(t,s)$, the one-sided Green's function associated with $\Theta_t(B)$, is such that $\sum\limits_{s=-\infty}^{t} |h(t,s)|^\delta <\infty$, for all $t$. Assume also that $\sum\limits_{s=-0}^{p} |\alpha_j(\cdot)|^2 < \infty$ for all $t$, and $\Phi_{t}(z)$ and $\Theta_{t}(z)$ $(\Theta_{t}(z) \neq 0$ for $|z|\leq 1)$ have no common roots. Then, the process \eqref{eq:tvARMA_stable_matrix} is invertible and its explicit inversion is given by 
	\begin{equation}
	\varepsilon_{t}= \sum_{j=0}^{\infty} b_{t,T}(j) X_{t-j,T}.
	\end{equation}
	where $X_{t,T}$ denotes an arbitrary solution and the coefficients are uniquely determined by 
	\begin{equation}
	b_{t,T}(j)=
	\left\{ \begin{array}{l r} 
	0, & j<0,\\
	\frac{1}{\gamma(\frac{t}{T})} , & j=0, \\
	\frac{1}{\gamma(\frac{t}{T})} \sum\limits_{l=0}^k \alpha_k (\frac{t-j+k}{T}) h(t,t-j+k), & 0\leq j \leq p, \\
	\frac{1}{\gamma(\frac{t}{T})} \sum\limits_{l=0}^q \alpha_k (\frac{t-j+k}{T}) h(t,t-j+k),  & j>p. 
	\end{array}	
	\right.
	\end{equation}
\end{theorem}


\begin{theorem}
	Let $\left\lbrace X_{t,T} \right\rbrace $ be a sequence of stochastic process that satisfies \eqref{eq:tvARMA_stable} that is AR regular. The solution $ X_{t,T} $ of the form \eqref{eq:stable_ma_inf} is strictly stable and ${X_{t,T} \sim S_\alpha(\sigma^*,\beta^*,0)}$, with
	$$
	\sigma^*=\left( \frac{1}{\sqrt{2}}\right) \left\lbrace \sum_{j=0}^{\infty} \left| a_{t,T}(j) \right|^\alpha  \right\rbrace ^{1/\alpha},\text{and}~~ \beta^*=\beta \left\lbrace  \frac{ \sum\limits_{j=0}^{\infty} \sign\left[ a_{t,T}(j)\right]  \left| a_{t,T}(j) \right|^\alpha}{ \sum\limits_{j=0}^{\infty} \left| a_{t,T}(j) \right|^\alpha }\right\rbrace .
	$$
\end{theorem}

\begin{proof}
	The explicit form of the solution is straightforward since the linear combination of stable distributions is also stable. Moreover, the Property 1.2.6 from \cite{Samorodnitsky1994} implies that for each $t$, the solution $X_{t,T}$ is strictly stable since each of them has location parameter equals $0$.  
\end{proof}

\subsection{Local Stationarity}

Similar to the Proposition 2.4 in \cite{Dahlhaus2009}, we can present the corresponding version for stable innovations. Since it is not a second-order process, the time-varying spectral density does not exist.

\begin{theorem} \label{theorem:tvARMA_stable} Consider the system of difference equations in \eqref{eq:tvARMA_stable} satisfying the AR regular conditions stated above. Suppose that all $\alpha_j(\cdot)$ and $\beta_k(\cdot)$, as well as $\gamma^2(\cdot)$ are of bounded variation. Then, there exists a solution of the form 
	$$
	X_{t,T}= \sum_{j=0}^{\infty} a_{t,T}(j) \varepsilon_{t-j},
	$$
	which fullfills \eqref{assumption1_iia}, \eqref{assumption1_iib} and \eqref{assumption1_iic}.
\end{theorem}	

\begin{proof}
	We give the proof for tvAR(p) process (i.e. $q=0$) and then the extension to tvARMA(p,q) is straightforward. Since the process \eqref{eq:tvARMA_stable} is AR regular, there exists a solution of the form
	$$
	X_{t,T}= \sum_{j=0}^{\infty} a_{t,T}(j) \varepsilon_{t-j},
	$$
	that is well defined and the coefficients are given by 
	\begin{equation*}
	a_{t,T}(j)= \left[ \prod_{\ell=0}^{j-1} \boldsymbol{\alpha}\left( \frac{t-\ell}{T}\right)  \right]_{11} \gamma \left( \frac{t-j}{T} \right) 
	\end{equation*}
	with
	$$
	\boldsymbol{\alpha}(u)=
	\begin{pmatrix} 
	-\alpha_1(u) & -\alpha_2(u) & \cdots & \cdots & -\alpha_p(u) \\ 
	1 & 0 & \cdots & \cdots & 0 \\
	0& \ddots & \ddots & & \vdots \\
	\vdots & \ddots & \ddots & \ddots & \vdots \\
	0& \cdots & 0 & 1 & 0	
	\end{pmatrix}
	$$
	\cite[for more detail see][]{Miller1968}. Then, the proof of the existence of the functions $\alpha(\cdot,j)$ satisfying \eqref{assumption1_iia}, \eqref{assumption1_iib} and \eqref{assumption1_iic} follows the same proof to that of finite innovation case (see Appendix in \cite{Dahlhaus2009}).	
\end{proof}

\begin{remark}
	\hfill
	\label{remark1}
	\begin{enumerate}
		\item Note that since $a_{t,T}(j)\approx a\left( \frac{t}{T},j\right)$, $X_{t,T}$ can be approximated by 
		\begin{equation}
		\label{eq:stable_ma_inf_aprox}
		\tilde{X}_{t,T}= \sum_{j=0}^{\infty} a\left( \frac{t}{T},j\right) \varepsilon_{t-j},
		\end{equation}
		which converges a.s. if and only if $\sum\limits_{j=0}^{\infty} \left| a\left( \frac{t}{T},j\right)\right|^\alpha< \infty.$
		
		Moreover, $\tilde{X}_{t,T} \sim S_\alpha(\sigma^+,\beta^+,0)$, with
		$$
		\sigma^+= \frac{1}{\sqrt{2}}\left\lbrace \sum_{j=0}^{\infty} \left| a\left( \frac{t}{T},j\right) \right|^\alpha  \right\rbrace ^{1/\alpha},\text{and}~~ \beta^+=\beta \left\lbrace  \frac{ \sum\limits_{j=0}^{\infty} \sign\left[ a\left( \frac{t}{T},j\right) \right]  \left| a\left( \frac{t}{T},j\right) \right|^\alpha}{ \sum\limits_{j=0}^{\infty} \left| a\left( \frac{t}{T},j\right) \right|^\alpha }\right\rbrace.
		$$
		\item $X_{t,T}$ in \eqref{eq:stable_ma_inf} can be expressed as a linear combination of $\alpha$-stable random variables and $X_{t,T}$ is strictly stable with the same index of stability $\alpha$. 
		\item Observe that $X_{t,T}$ is not strictly stationary, but it can be approximated by $\tilde{X}_{t,T}$ which is locally (strictly) stationary and strictly stable with the same index of stability.
		\item Weak stationarity does not make sense since the second moment does not exist. Consequently, the time-varying spectral representation does not exist.
		\item Let $X_{1,T},...,X_{T,T}$ be the sequence of solutions defined in \eqref{eq:stable_ma_inf} and $\tilde{X}_{1,T},...,\tilde{X}_{T,T}$ be the sequence of the stochastic process defined in \eqref{eq:stable_ma_inf_aprox}. Both processes are strictly $\alpha-$stable, since all linear combinations are strictly stable with the same index of stability. This means that the weak stationarity is lost but it is substituted by the same tail behavior throughout the time. This is the reason we call this process $\alpha-$stable locally (strictly) stationary process. 
	\end{enumerate}
	
\end{remark}

If we consider the symmetric $\alpha-$stable ($S\alpha S$) innovations, i.e. $\beta=0$, the simplest form is obtained.

\begin{corollary}[tvARMA with symmetric stable innovations]
	Let $X_{t,T}$ be a sequence of stochastic process that satisfies \eqref{eq:tvARMA_stable} with i.i.d. $S\alpha S$ innovations, that is, $\varepsilon_t \sim S_\alpha ( \nicefrac{1}{\sqrt{2}},0,0)$ with $\alpha \in (0,2)$. Then, there exists a solution of the form \eqref{eq:stable_ma_inf}. This solution $ X_{t,T} $ is symmetric, $\alpha-$stable and $X_{t,T} \sim S_\alpha(\sigma^*,0,0)$, with
	$$
	\sigma^*= \frac{1}{\sqrt{2}} \left\lbrace \sum_{j=0}^{\infty} \left| a_{t,T}(j) \right|^\alpha  \right\rbrace ^{1/\alpha}.
	$$
	Similarly to the general case, $X_{t,T}$ can be approximated by $\tilde{X}_{t,T} \sim S_\alpha(\sigma^+,0,0)$ as in \eqref{eq:stable_ma_inf_aprox}, with
	$$
	\sigma^+= \left\lbrace \sum_{j=0}^{\infty} \left| a\left( \frac{t}{T},j\right) \right|^\alpha  \right\rbrace ^{1/\alpha}.
	$$
\end{corollary}

In the symmetric case, in addition to the properties in Remark \ref{remark1}, note that the processes $X_{t,T}$ and $\tilde{X}_{t,T}$ are symmetric $\alpha-$stable. 	

\subsection{Some examples}

\begin{example}	
	The tvMA(q) model with stable innovations:
	\begin{equation}
	\label{eq:tvMAq}
	X_{t,T}=\sum_{k=0}^{q} \beta_k\left( \frac{t}{T} \right)\gamma \left( \frac{t-k}{T} \right)  \varepsilon_{t-k}.
	\end{equation}	
\end{example}

\begin{example}	
	\normalfont
	Consider the tvAR(p) model with stable innovations
	\begin{equation}
	\label{eq:tvARp}
	\sum_{j=0}^{p} \alpha_j\left( \frac{t}{T} \right)  X_{t-j,T} = \gamma\left( \frac{t}{T} \right) \varepsilon_{t}.
	\end{equation}
	Underregularity conditions, $X_{t,T}$ does not have a solution of the form
	$$
	X_{t,T}=\sum\limits_{k=0}^{\infty} a_k \left( \frac{t}{T} \right) \varepsilon_{t-k},
	$$
	but only of the form \eqref{eq:stable_ma_inf} with 
	\begin{equation}
	a_{t,T}(j)= \left[ \prod_{\ell=0}^{j-1} \boldsymbol{\alpha}\left( \frac{t-\ell}{T}\right)  \right]_{11} \gamma \left( \frac{t-j}{T} \right) 
	\end{equation}
	 and
	$$
	\boldsymbol{\alpha}(u)=
	\begin{pmatrix} 
	-\alpha_1(u) & -\alpha_2(u) & \cdots & \cdots & -\alpha_p(u) \\ 
	1 & 0 & \cdots & \cdots & 0 \\
	0& \ddots & \ddots & & \vdots \\
	\vdots & \ddots & \ddots & \ddots & \vdots \\
	0& \cdots & 0 & 1 & 0	
	\end{pmatrix}
	$$ where $\boldsymbol{\alpha}(u)=\boldsymbol{\alpha}(0)$ for $u<0$ \cite{Dahlhaus2009}. Moreover, $a_{t,T}(j)$ can be approximated by $a(u,j)=\left( \boldsymbol{\alpha}(u)^j\right)_{11} \gamma(u) $ which satisfies \eqref{assumption1_iia}, \eqref{assumption1_iib} and \eqref{assumption1_iic}.
	
	Figure \ref{fig:tvAR(1)} presents simulated tvAR(1) process of $T=1000$ observations with different innovation distribution (Gaussian, $t_{(3)}$ and symmetric stable innovations, $\alpha=1.7,1.4,0.9,0.6$) and a linear coefficient $\alpha_1(u)=-0.2+0.6u$ and $\gamma(u)=1$. We observe that for smaller $\alpha$, the process seems to have more outliers.
\end{example}
	
	\begin{figure}[!htbp]
	\centering
	\includegraphics[width=\textwidth]{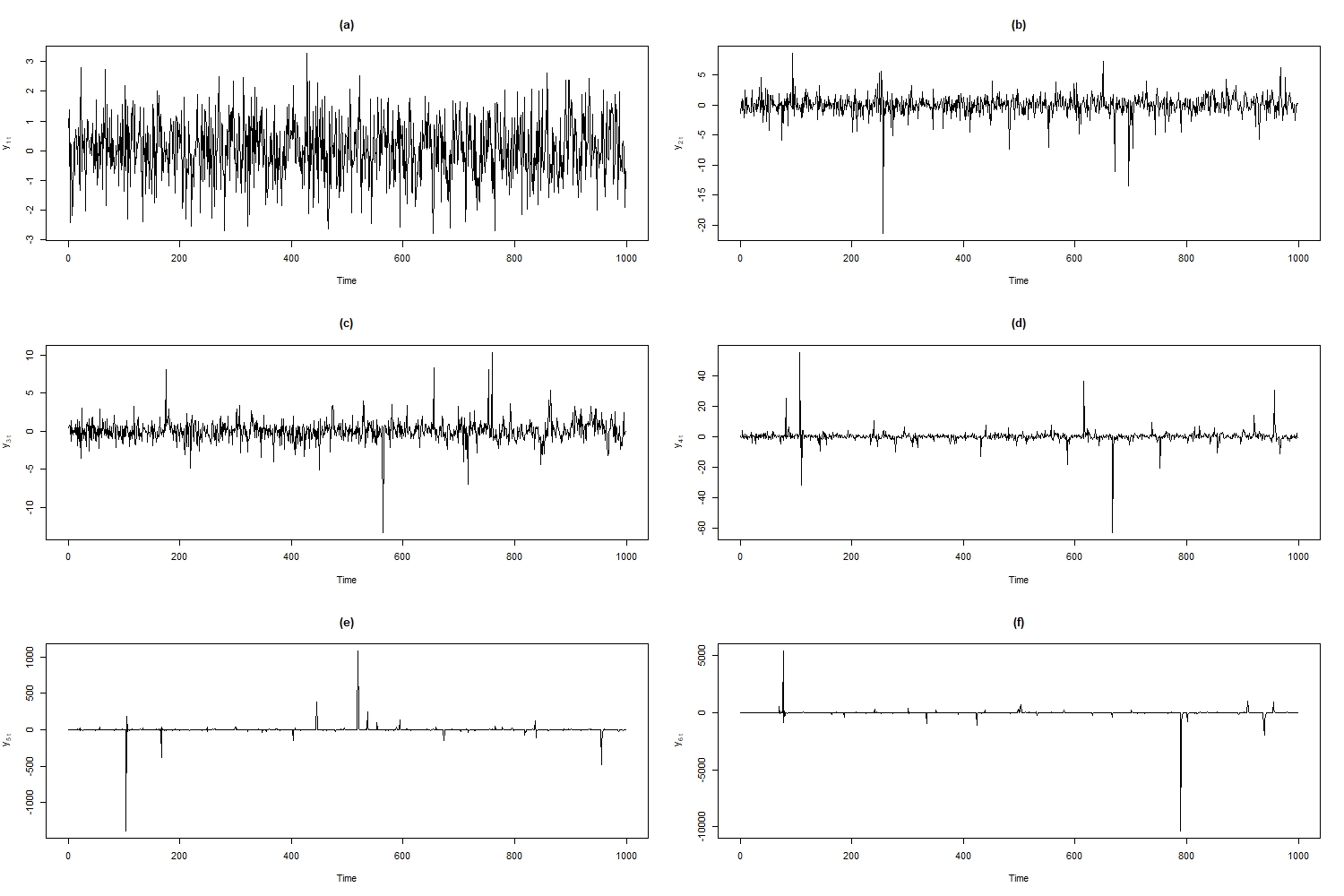}
	\caption{Simulated tvAR(1) assuming time-varying coefficient $\alpha_1(u)=-0.2+0.6u$ and $\gamma(u)=1$, with (a) Gaussian, (b) $t_{(3)}$ innovations, and symmetric stable innovations with (c) $\alpha=1.7$, (d) $\alpha=1.4$, (e) $\alpha=0.9$ and (f) $\alpha=0.6$.} 
	\label{fig:tvAR(1)}
\end{figure}

\subsection{Prediction}

Recalling that $\alpha-$stable tvARMA has infinite variance, prediction results based on stable ARMA processes with dependent coefficients are presented by \cite{Peiris2001a,Peiris2001b}. Then, it is possible to predict future values along with the approach applied by \cite{VanBellegem2004}, which considers the observed values $X_{0,T},\cdots,X_{T-h-1,T}$ and rescaling the time interval to $ \left[0, 1-\frac{h+1}{T} \right]$, where $h$ is the forecasting horizon and the ratio $h/T$ tends to zero as $T$ tends to infinity. Here, we consider that the innovations are $S \alpha S$ random variables. 

Suppose that we have the system of difference equation \eqref{eq:tvARMA_stable} that satisfies above regular conditions, and $X_{0,T},\cdots,X_{T',T}$ with $T'=T-h-1$ are observed. We are interested in predictions with horizon $h$, i.e. $X_{T-h,T},...,X_{T,T}$.

Since $X_{t,T}$ is AR regular, it can be expressed as
\begin{equation}
X_{t,T}= \sum_{j=0}^{\infty} a_{t,T}(j) \varepsilon_{t-j}.
\end{equation}
Let $\hat{X}_{T'}(l)$ be the best linear predictor of $X_{T'+l,T}$ for $l=1,...,h$, namely
\begin{equation}
\hat{X}_{T'}(l)= \sum\limits_{j=0}^{\infty} A(T',T'-j)\varepsilon_{T'-j},
\end{equation}
where $A(T',T'-j)$ are some functions. Since the prediction error $e_{T'}(l)=\hat{X}_{T'+h,T}-X_{T'}(l)$ is also $S \alpha S$ random variable, it is possible to define its dispersion as $d=\sigma^\alpha$ with $\sigma$ its scale parameter. The idea is to minimize the dispersion $d$. Note that
\begin{equation}	
\begin{array}{l l} 
e_{T'}(l) & =X_{T'+l,T}-\hat{X}_{T'}(l) \\
& = \sum\limits_{j=0}^{\infty} a_{T'+l,T}(j)\varepsilon_{T'+l-j} - \sum\limits_{j=0}^{\infty} A(T',T'-j)\varepsilon_{T'-j} \\
& = \sum\limits_{j=0}^{l-1} a_{T'+l,T}(j)\varepsilon_{T'+l-j}+ \sum\limits_{j=l}^{\infty} a_{T'+l,T}(j)\varepsilon_{T'+l-j}- \sum\limits_{j=0}^{\infty} A(T',T'-j)\varepsilon_{T'-j}\\
& =\sum\limits_{j=0}^{l-1} a_{T'+l,T}(j)\varepsilon_{T'+l-j}+ \sum\limits_{j=0}^{\infty} \left(  a_{T'+l,T}(j+l)- A(T',T'-j)\right)  \varepsilon_{T'-j}.
\end{array}	
\end{equation}
Then, assuming $\varepsilon_{t} \sim S_\alpha(\frac{1}{\sqrt{2}},0,0)$ and using properties of $S \alpha S$ random variables, its dispersion is
\begin{equation}
\label{eq:disp_error}
\text{disp}\left[ e_{T'}(l)\right] =\left( \frac{1}{\sqrt{2}}\right)^\alpha \sum\limits_{j=0}^{l-1} |a_{T'+l,T}(j)|^\alpha+ \left( \frac{1}{\sqrt{2}}\right)^\alpha \sum\limits_{j=0}^{\infty} |\left(  a_{T'+l,T}(j+l)- A(T',T'-j)\right)|^\alpha.
\end{equation}
Minimizing the expression \eqref{eq:disp_error}, we obtain the following theorem.
\begin{theorem}
	The minimum dispersion predictor is given by
	\begin{equation}
	\hat{X}_{T'}(l)= \sum\limits_{j=0}^{\infty} a_{T'+l,T}(j+l)\varepsilon_{T'-j}.
	\end{equation}
\end{theorem}

\begin{proof}
	From \eqref{eq:disp_error}, it is straightforward to obtain
	$$
	\text{min disp}\left[ e_{T'}(l)\right]=\left( \frac{1}{\sqrt{2}}\right)^\alpha \sum\limits_{j=0}^{l-1} |a_{T'+l,T}(j)|^\alpha,
	$$
	with  $a_{T'+l,T}(j+l)= A(T',T'-j)$ for $j=0,1,...$.
\end{proof}

\section{Indirect inference for $\alpha-$stable tvARMA processes}
\label{sec:indirect_inference_stable_tvARMA}

The IM is the tvARMA with innovation $\varepsilon_t \sim S_\alpha (\nicefrac{1}{\sqrt{2}},\beta,0)$. We study the parameter estimation when the innovation parameters $\alpha$ and $\beta$ are known. Then, we study the case assuming unknown $\alpha$.

Suppose that the parameter curves of model of interest can be parametrized by a finite-dimensional parameter $\theta$. The estimation strategy is to consider an auxiliary model with the same parametric time varying coefficient structure with Student's t innovations. The conditional likelihood estimates, defined in the equation $(30)$ in \cite{Dahlhaus2012}, is,

\begin{equation}
\label{eq:parametric_fit}
\boldsymbol{\hat{\lambda}}= \operatorname*{argmin}_{\boldsymbol{\lambda}} \frac{1}{T} \sum_{t=1}^{T}  \ell_{t,T}\left( \theta_\lambda \left( \frac{t}{T} \right) \right),
\end{equation}
where $\ell_{t,T}(\boldsymbol{\theta})=-\log f_\theta(X_{t,T}|X_{t-1,T},...,X_{1,T})$. In practice, we will use t distribution with $\nu=3$ degrees of freedom for the case of known parameters since its tail is heavier than the Gaussian one. For unknown $\alpha$ case, we let $\nu$ to be estimated in the AM.

\section{Simulation study}
\label{sec:simulation}

This section presents a Monte Carlo (MC) simulation in order to investigate the properties of the indirect inference estimators. All the simulation programs and routines were implemented in R language. We present one scenario for each of the following models but still different values of $\alpha$ were selected. Some other scenarios were performed for each case and similar results were obtained and they are available upon request from the authors. For each scenario, simulations were done for $T=500$, $1000$ and $1500$ observations based on $R=1000$ independent replications. The indirect inference was carried out using $S=100$. For $\alpha$ known, we also performed the blocked Whittle estimation (BWE), proposed by \cite{Dahlhaus1997}, to compare the estimation of time structure of the model with indirect inference. The suggestion of block size $N=\floor{T^{0.8}}$ and shifting each block by $S=\floor{0.2 N}$ time units from \cite{Dahlhaus1998} is used.

\subsection{Known $\alpha$ case}
\label{sec:known_alpha}
\subsubsection{$\alpha$-stable tvAR(1)}

Consider tvAR(p) in \eqref{eq:tvARp} with $p=1$ and $\gamma\left( \frac{t}{T} \right)=\gamma$:
\begin{equation}\label{eq:tvAR1}
X_{t,T}+\alpha_1 \left( \frac{t}{T} \right) X_{t-1,T} =\gamma \varepsilon_{t},
\end{equation}
where $\varepsilon_{t} \sim S_\alpha \left( \nicefrac{1}{\sqrt{2}},\beta ,0 \right)$ with $\alpha$ and $\beta$ known.

We illustrate how the indirect inference can be employed to the tvAR(1) with the linear parametric form of the time varying coefficient $\alpha_1 \left(  u \right) =\theta_0+\theta_1 u$, and we consider that $\varepsilon_{t} \sim S_\alpha(\nicefrac{1}{\sqrt{2}},\beta,0)$ for $\alpha$ known. Therefore, the parameters of the IM is $\theta=\left( \theta_0,\theta_1, \gamma \right)$. 

The simulation was performed by assuming known parameters $\alpha=1.9$ and $\beta=0.9$ and unknown $(\theta_0,\theta_1,\gamma)=(-0.3,0.8,1)$. It is important to report that since $\alpha$ is close to $2$, all replications for the BWE converged. This outcome is expected since the innovation distributions approximate to the Gaussian distribution for $\alpha$ close to $2$. 

Table \ref{tab:simulation_tvAR1} reports the MC mean and standard error of both estimation methods. Notice that the MC mean from the indirect estimates seems to be consistent, that is, they approximate the real parameters and present lower standard errors as $T$ increases. On the other hand, the MC mean of the BWE are different from the real parameters and present higher standard errors compared to our estimation approach. 

Table \ref{tab:simulation_tvAR1_skew_kur} presents the kurtosis and skewness of all estimates from both methods. In general, all indirect estimates present lower kurtosis and the skewness close to $0$. Notice that since the second moment of the process does not exist, the parameter $\gamma$ estimates from the BWE present highly positive asymmetry and they subestimate the true parameter.

\begin{table}[!htbp]
	\centering
		\caption{MC means and standard errors for different sample size $T$, using indirect estimators and BWE assuming  $\alpha$ and $\beta$ known, from $\alpha$-stable tvAR(1) with $(\alpha,\beta,\theta_0,\theta_1,\gamma)=(1.9,0.9,-0.3,0.8,1),$  based on $R=1000$ replications.}
		\label{tab:simulation_tvAR1}
		\begin{tabular}{c cccccc}
			\toprule
			\multirow{2}{*}{$T$} & \multicolumn{3}{c}{Indirect estimates} & 
			\multicolumn{3}{c}{BWE} \\  \cline{2-7}
			& $\theta_0$ & $\theta_1$ & $\gamma$ & $\theta_0^{(W)}$ & $\theta_1^{(W)}$ & $\gamma^{(W)}$ \\ \hline
			\multirow{2}{*}{$500$} & -0.2952 & 0.7897 & 0.9966 & -0.2880 & 0.7825 & 1.2086 \\ 
			& (0.0881) & (0.1523) & (0.0366) & (0.1172) & (0.2216) & (0.6352) \\ 
			\multirow{2}{*}{$1000$} & -0.2975 & 0.7926 & 0.9996 & -0.2917 & 0.7845 & 1.2197 \\ 
			& (0.0585) & (0.1028) & (0.0260) & (0.0811) & (0.1545) & (0.4734) \\ 
			\multirow{2}{*}{$1500$} & -0.2974 & 0.7958 & 0.9997 & -0.2940 & 0.7926 & 1.2709 \\ 
			& (0.0494) & (0.0793) & (0.0209) & (0.0639) & (0.1162) & (0.8738) \\ 
			\bottomrule
		\end{tabular}
\end{table}

\begin{table}[!htbp]
	\centering
		\caption{Kurtosis and skewness of indirect estimates and BWE for different sample size $T$, assuming  $\alpha$ and $\beta$ known, from $\alpha$-stable tvAR(1) with $(\alpha,\beta,\theta_0,\theta_1,\gamma)=(1.9,0.9,-0.3,0.8,1)$, based on $R=1000$ replications.}
		\label{tab:simulation_tvAR1_skew_kur}
		\begin{tabular}{c ccccccc}
			\toprule
			\multirow{2}{*}{$T$} & & \multicolumn{3}{c}{Indirect estimates} & 
			\multicolumn{3}{c}{BWE} \\  \cline{3-8}
			& & $\theta_0$ & $\theta_1$ & $\gamma$ & $\theta_0^{(W)}$ & $\theta_1^{(W)}$ & $\gamma^{(W)}$ \\ \hline
			\multirow{2}{*}{$500$} & Kur & 3.0330 & 2.8565 & 3.1076 & 3.3783 & 3.0944 & 375.8563 \\ 
			& Skw & 0.1388 & -0.1354 & 0.1241 & -0.0129 & -0.0875 & 16.6778 \\ 
			\multirow{2}{*}{$1000$} & Kur &  3.1678 & 3.2835 & 2.7390 & 2.8722 & 2.9543 & 95.2261 \\ 
			& Skw &  0.0341 & -0.0260 & 0.0437 & 0.0057 & -0.1047 & 8.4487 \\ 
			\multirow{2}{*}{$1500$} & Kur &  3.1024 & 3.0645 & 2.9329 & 3.7487 & 6.1799 & 187.9560 \\ 
			& Skw & -0.0299 & 0.0248 & 0.0026 & 0.1707 & -0.5935 & 12.4661 \\ 
			\bottomrule
		\end{tabular}
\end{table}

Figure \ref{fig:tvAR_1_alpha1.9_density} shows the density estimates of each parameter. They show that the standard error become smaller as $T$ increases. Along with the results from Tables \ref{tab:simulation_tvAR1} and \ref{tab:simulation_tvAR1_skew_kur}, we can conclude that indirect estimates behave better than the BWE in terms of mean, standard error, skewness and kurtosis. Therefore, the simulation results show that the indirect inference performs well.

\begin{figure}[!htbp]
	\centering
	\includegraphics[width=\textwidth]{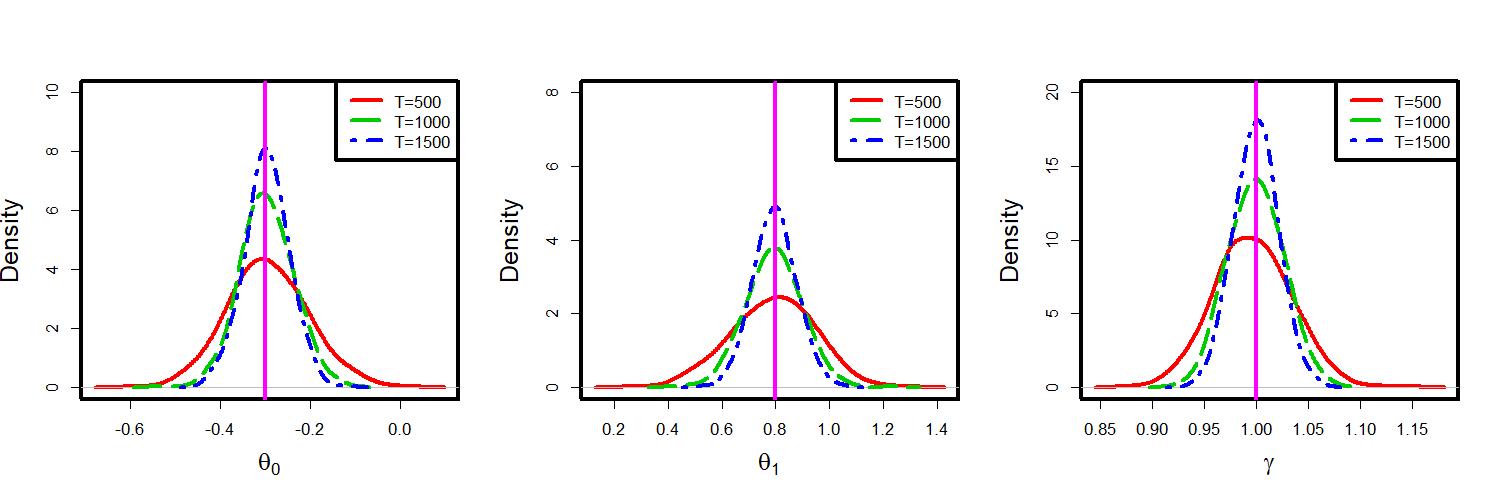}
	\caption{Density estimates of $\theta_0$, $\theta_1$ and $\gamma$ for different sample sizes, based on $R=1000$ replications from $\alpha$-stable tvAR(1) with $(\alpha,\beta,\theta_0,\theta_1,\gamma)=(1.9,0.9,-0.3,0.8,1)$, using indirect inference.} 
	\label{fig:tvAR_1_alpha1.9_density}
\end{figure}

\subsubsection{$\alpha$-stable tvMA(1)}

In this section, we carried out simulations for  a tvMA(q) in \eqref{eq:tvMAq} with $q=1$ and $\gamma\left( \frac{t}{T} \right)=\gamma$:
\begin{equation}\label{eq:tvMA1}
X_{t,T}=\gamma \left\lbrace  \varepsilon_{t} + \beta_1 \left( \frac{t}{T} \right)  \varepsilon_{t-1} \right\rbrace,
\end{equation}
where $\varepsilon_{t} \sim S_\alpha \left( \nicefrac{1}{\sqrt{2}},\beta,0 \right)$ with  $\alpha$ and $\beta$ known.

The indirect inference is employed for the linear parametric form of the time varying coefficient $\beta_1 \left(  u \right) =\theta_0+\theta_1 u$, and we consider that $\varepsilon_{t} \sim S_\alpha(\nicefrac{1}{\sqrt{2}},\beta,0)$ for known $\alpha$ and $\beta$. Hence, the vector of parameters of the model of interest is $\theta=\left( \theta_0,\theta_1, \gamma \right)$. 

This scenario assumes known $\alpha=1.1$ and $\beta=-0.2$ and unknown $(\theta_0,\theta_1,\gamma)=(0.35,-0.6,1.2)$. For the BWE case, we consider only $R= 939, 978$ and $978$ replications with converged estimates for $T=500, 1000$, and $1500$, respectively. This result is expected because BWE assumes finite second moment. The MC mean, standard error, kurtosis and skewness of estimates from the simulation are reported in the Table \ref{tab:simulation_tvMA1} and \ref{tab:simulation_tvMA1_skew_kur} and the density estimates in Figures \ref{fig:tvMA_1_alpha1.1_density}.

Similarly to the previous case, the indirect estimates seem to be consistent and the standard error become smaller as $T$ increases. For this case, since $\alpha$ is smaller, the distribution of indirect estimates has heavier tails, and they have similar kurtosis and skewness than  the BWE estimates, except for the parameter $\gamma$, when the indirect estimation behaves better. In addition, in term of standard error and MC mean, they still behave better than the BWE. We conclude that the indirect inference has a good performance. 

\begin{table}[!htbp]
	\centering
	\begin{threeparttable}
		\caption{MC mean and standard error for different sample size $T$, using indirect estimators and BWE assuming $\alpha$ and $\beta$ known, from $\alpha$-stable tvMA(1) with $(\alpha,\beta,\theta_0,\theta_1,\gamma)=(1.1,-0.2,0.35,-0.6,1.2)$, based on $R=1000$ replications.}
		\label{tab:simulation_tvMA1}
		\begin{tabular}{c ccc|ccc}
			\toprule
			\multirow{2}{*}{$T$} & \multicolumn{3}{c}{Indirect estimates} & 
			\multicolumn{3}{c}{BWE\footnotemark[1]} \\  \cline{2-7}
			& $\theta_0$ & $\theta_1$ & $\gamma$ & $\theta_0^{(W)}$ & $\theta_1^{(W)}$ & $\gamma^{(W)}$ \\ \hline
			\multirow{2}{*}{$500$} & 0.3561 & -0.5888 & 1.1989 & 0.3424 & -0.5427 & 18.7932 \\ 
			& (0.0298) & (0.0577) & (0.0600) & (0.1418) & (0.3084) & (38.4343) \\ 
			\multirow{2}{*}{$1000$} & 0.3545 & -0.5953 & 1.1986 & 0.3386 & -0.5532 & 47.5752 \\ 
			& (0.0186) & (0.0352) & (0.0412) & (0.0870) & (0.1955) & (232.0620) \\ 
			\multirow{2}{*}{$1500$} & 0.3536 & -0.5982 & 1.1986 & 0.3357 & -0.5555 & 49.4572 \\ 
			& (0.0131) & (0.0244) & (0.0331) & (0.0747) & (0.1690) & (178.3984) \\
			\bottomrule
		\end{tabular}
	\end{threeparttable}
\end{table}
\footnotetext[1]{In tvMA(1) simulations, the BWE did not converge in some cases. Therefore, excluding those cases, $R= 939, 978$ and $978$ replications are included for $T=500, 1000$, and $1500$, respectively.}

\begin{table}[!htbp]
	\centering
	\begin{threeparttable}
		\caption{Kurtosis and skewness of indirect estimates and BWE for different sample size $T$ assuming known $\alpha$ and $\beta$ from $\alpha$-stable tvMA(1) with $(\alpha,\beta,\theta_0,\theta_1,\gamma)=(1.1,-0.2,0.35,-0.6,1.2)$ based on $R=1000$ replications.}
		\label{tab:simulation_tvMA1_skew_kur}
		\begin{tabular}{c ccccccc}
			\toprule
			\multirow{2}{*}{$T$} & & \multicolumn{3}{c}{Indirect estimates} & 
			\multicolumn{3}{c}{BWE\footnotemark[1]} \\  \cline{3-8}
			& & $\theta_0$ & $\theta_1$ & $\gamma$ & $\theta_0^{(W)}$ & $\theta_1^{(W)}$ & $\gamma^{(W)}$ \\ \hline
			\multirow{2}{*}{$500$} & Kur & 7.9023 & 6.3460 & 2.9050 & 6.9952 & 7.3434 & 233.1652 \\ 
			& Skw & 1.2950 & 0.0800 & 0.2117 & 0.1961 & 1.1869 & 13.5324 \\ 
			\multirow{2}{*}{$1000$} & Kur &  9.8633 & 10.4926 & 2.8616 & 11.7454 & 8.6755 & 385.9194 \\ 
			& Skw & 1.6510 & 0.7121 & 0.0841 & 0.8633 & 1.5096 & 17.6374 \\ 
			\multirow{2}{*}{$1500$} & Kur &  8.1466 & 20.6873 & 2.8140 & 9.7011 & 10.5014 & 156.1843 \\ 
			& Skw & 1.5194 & 1.4762 & 0.1493 & -0.1837 & 1.9926 & 11.4943 \\
			\bottomrule
		\end{tabular}
	\end{threeparttable}
\end{table}

\begin{figure}[!htbp]
	\centering
	\includegraphics[width=\textwidth]{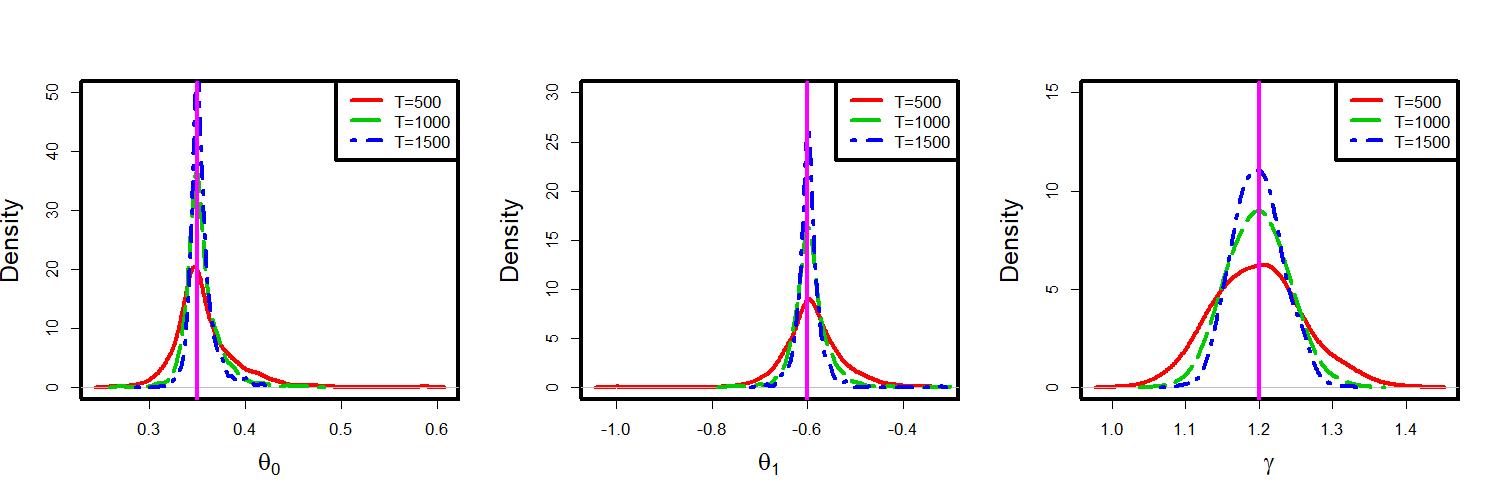}
	\caption{Density estimates of $\theta_0$, $\theta_1$ and $\gamma$ for different sample sizes based on $R=1000$ replications from $\alpha$-stable tvMA(1) with $(\alpha,\beta,\theta_0,\theta_1,\gamma)=(1.1,-0.2,0.35,-0.6,1.2)$ using indirect inference.} 
	\label{fig:tvMA_1_alpha1.1_density}
\end{figure}

\subsubsection{$\alpha$-stable tvARMA(1,1)}

The third simulation was carried out with the case of tvARMA(p,q) with $p=1$, $q=1$ and $\gamma\left( \frac{t}{T} \right)=\gamma$:

\begin{equation}\label{eq:tvARMA1}
X_{t,T}+\alpha_1 \left( \frac{t}{T} \right) X_{t-1,T}=\gamma \left\lbrace  \varepsilon_{t} + \beta_1 \left( \frac{t}{T} \right)  \varepsilon_{t-1} \right\rbrace,
\end{equation}
where $\varepsilon_{t} \sim S_\alpha \left( \nicefrac{1}{\sqrt{2}},\beta,0 \right)$ with $\alpha$ and $\beta$ known.

We suppose a linear parametric form of the time varying coefficients $\alpha_1 \left(  u \right) =\theta_{a0}+\theta_{a1} u$ and $\beta_1 \left(  u \right) =\theta_{b0}+\theta_{b1} u$. Therefore, the parameters of the IM is $\theta=\left( \theta_{a0},\theta_{a1}, \theta_{b0},\theta_{b1}, \gamma \right)$.

The simulation was done by assuming $\alpha=1.8$, $\beta=0.3$ and $\left( \theta_{a0},\theta_{a1}, \theta_{b0},\theta_{b1}, \gamma \right)= (-0.4,0.1,0.1,0.3,1.1)$. For BWE, $R= 989, 996$ and $994$ replications with converged estimates are included for $T=500, 1000$, and $1500$, respectively.

The MC mean, standard error, kurtosis and skewness of estimates from the tvARMA(1,1) simulation are reported in the Table \ref{tab:simulation_tvARMA1} and \ref{tab:simulation_tvARMA1_skew_kur} and the density estimates in Figure \ref{fig:tvARMA_1_alpha1.8_density}. In general, the distribution of indirect estimates has heavier tails, and the kurtosis and skewness are similar to the BWE (except for the parameter $\gamma$, indirect estimates behave better). However, in terms of standard error and MC mean, they behave much better than the BWE. Therefore, the indirect inference works well for tvARMA(1,1). 

\begin{table}[!htbp]
	
	\centering
		\caption{MC mean and standard error for different sample sizes $T$, using indirect estimators and BWE assuming  $\alpha$ and $\beta$ known, from $\alpha$-stable tvARMA(1,1) with $(\alpha,\beta,\theta_{a0},\theta_{a1},\theta_{b0},\theta_{b1},\gamma)=(1.8,0.3,-0.4,0.1,0.1,0.3,1)$, based on $R=1000$ replications.}
		\label{tab:simulation_tvARMA1}
		\scriptsize
		\begin{tabular}{c cccccccccc}
			\toprule
			\multirow{2}{*}{$T$} & \multicolumn{5}{c}{Indirect estimates} & 
			\multicolumn{5}{c}{BWE\footnotemark[2]} \\ \cline{2-11}
			& $\theta_{a0}$ & $\theta_{a1}$ & $\theta_{b0}$ & $\theta_{b1}$ & $\gamma$ & $\theta_{a0}^{(W)}$ & $\theta_{a1}^{(W)}$ & $\theta_{b0}^{(W)}$ & $\theta_{b1}^{(W)}$ & $\gamma^{(W)}$ \\ \hline
			\multirow{2}{*}{$500$}  & -0.4000 & 0.1061 & 0.0987 & 0.3097 & 0.9976 & -0.3917 & 0.1021 & 0.1078 & 0.3049 & 1.4151 \\ 
			& (0.1360) & (0.2222) & (0.1501) & (0.2395) & (0.0386) & (0.1952) & (0.3522) & (0.2130) & (0.3810) & (0.7603) \\ 
			\multirow{2}{*}{$1000$}  & -0.3921 & 0.0881 & 0.1064 & 0.2905 & 0.9982 & -0.3850 & 0.0815 & 0.1105 & 0.2880 & 1.4919 \\ 
			& (0.1001) & (0.1617) & (0.1053) & (0.1652) & (0.0290) & (0.1409) & (0.2535) & (0.1470) & (0.2599) & (0.5806) \\ 
			\multirow{2}{*}{$1500$}  & -0.3992 & 0.1021 & 0.0988 & 0.3060 & 0.9982 & -0.3939 & 0.0926 & 0.1055 & 0.2964 & 1.5538 \\ 
			& (0.0754) & (0.1269) & (0.0793) & (0.1285) & (0.0232) & (0.1085) & (0.1955) & (0.1155) & (0.2040) & (0.8009) \\ 
			\bottomrule
		\end{tabular}
\end{table}

\begin{table}[!htbp]

	\centering
		\caption{Kurtosis and skewness of indirect estimates and BWE for different sample size $T$, assuming known $\alpha$ and $\beta$ from $\alpha$-stable tvARMA(1,1) with $(\alpha,\beta,\theta_{a0},\theta_{a1},\theta_{b0},\theta_{b1},\gamma)=(1.8,0.3,-0.4,0.1,0.1,0.3,1)$, based on $R=1000$ replications.}
		\label{tab:simulation_tvARMA1_skew_kur}
			\scriptsize
		\begin{tabular}{c ccccccccccc}
			\toprule
			\multirow{2}{*}{$T$} & & \multicolumn{5}{c}{Indirect estimates} & \multicolumn{5}{c}{BWE\footnotemark[2]} \\ \cline{3-12}
			&  &  $\theta_{a0}$ & $\theta_{a1}$ & $\theta_{b0}$ & $\theta_{b1}$ & $\gamma$ & $\theta_{a0}^{(W)}$ & $\theta_{a1}^{(W)}$ & $\theta_{b0}^{(W)}$ & $\theta_{b1}^{(W)}$ & $\gamma^{(W)}$ \\ \hline
			\multirow{2}{*}{$500$} & Kur & 3.3650 & 3.1791 & 3.3657 & 3.4297 & 2.9935 & 2.9112 & 3.0692 & 3.2168 & 3.2822 & 203.2823 \\ 
			& Skw & 0.2754 & -0.2426 & -0.0746 & -0.1274 & 0.1839 & 0.1699 & -0.1329 & -0.2024 & -0.0422 & 12.0737 \\ 
			\multirow{2}{*}{$1500$} & Kur & 3.3964 & 3.5054 & 3.4470 & 3.4690 & 3.0327 & 3.4242 & 3.2166 & 3.2601 & 3.0064 & 41.6803 \\ 
			& Skw & 0.2002 & -0.1558 & 0.0100 & -0.1184 & 0.2460 & 0.3149 & -0.2253 & 0.0267 & -0.1713 & 4.9608 \\ 
			\multirow{2}{*}{$1500$} & Kur & 3.5817 & 3.1935 & 3.7052 & 3.3930 & 2.9790 & 2.9176 & 2.8801 & 3.3685 & 3.3091 & 96.2273 \\ 
			& Skw & 0.2895 & -0.1097 & 0.0137 & -0.0730 & 0.0718 & 0.0809 & 0.0268 & -0.1083 & 0.0372 & 7.7396 \\ 
			\bottomrule
		\end{tabular}
\end{table}

\footnotetext[2]{In tvARMA(1,1) simulations, the BWE did not converge in some cases. Therefore, excluding those cases, $R= 989, 996$ and $994$ replications are included for $T=500, 1000$, and $1500$, respectively.}

\begin{figure}[!htbp]
	\centering
	\includegraphics[width=\textwidth]{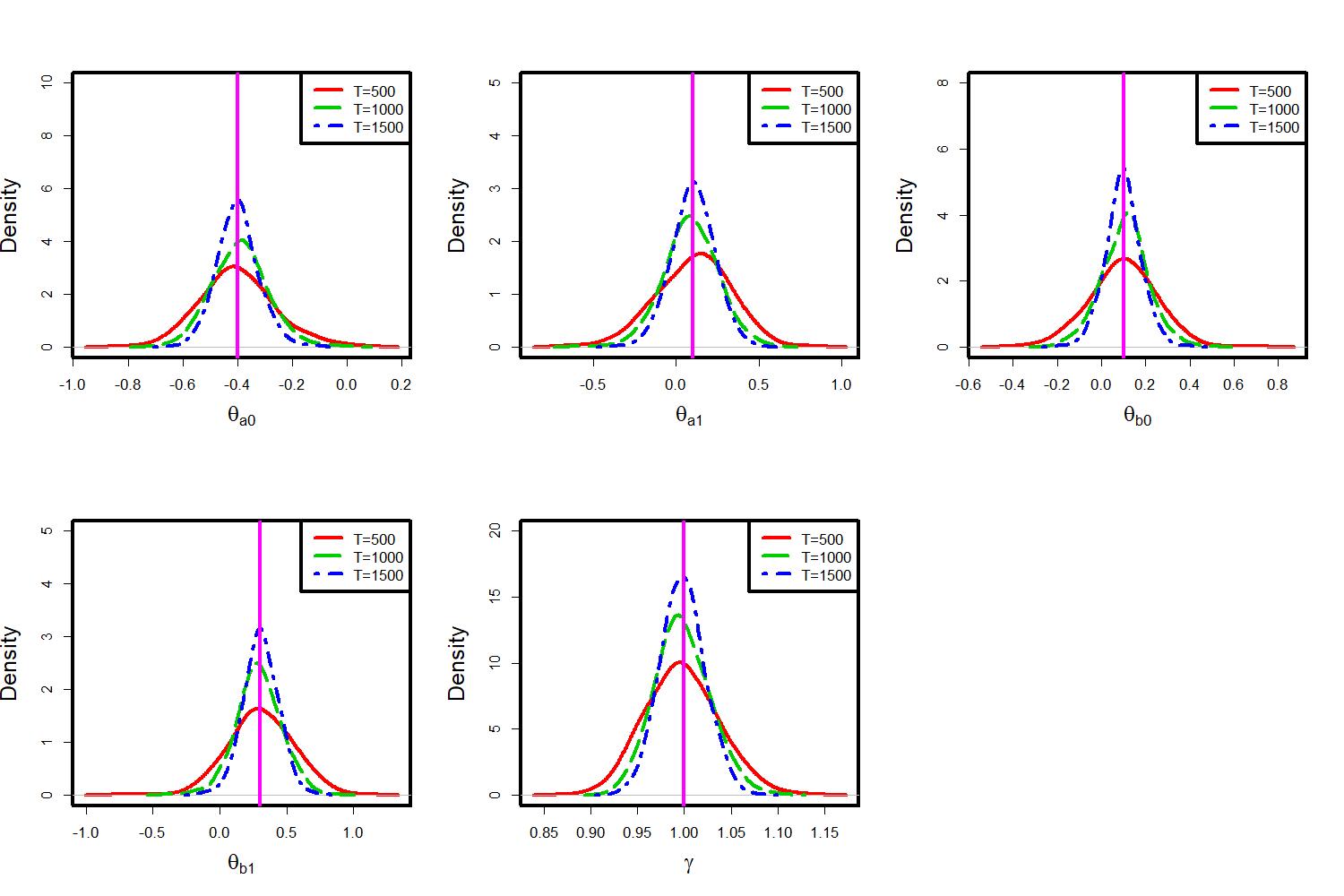}
	\caption{Density estimates of $\theta_{a0}$, $\theta_{a1}$, $\theta_{b0}$, $\theta_{b1}$ and $\gamma$ for different sample sizes based on $R=1000$ replications from $\alpha$-stable tvARMA(1,1) with $(\alpha,\beta,\theta_{a0}, \theta_{a1}, \theta_{b0}, \theta_{a1},\gamma)=(1.8,0.3,-0.4,0.1,0.1,0.3,1)$ using indirect inference.} 
	\label{fig:tvARMA_1_alpha1.8_density}
\end{figure}

\subsection{Unknown $\alpha$ case}
\label{sec:unknown_alpha}
\subsubsection{$\alpha$-stable tvAR(1)}

Consider the tvAR(1) model
\begin{equation}\label{eq:tvAR1_unknown}
X_{t,T}+\alpha_1 \left( \frac{t}{T} \right) X_{t-1,T} =\gamma \left( \frac{t}{T} \right) \varepsilon_{t},
\end{equation}
where $\varepsilon_{t} \sim S_\alpha \left( \nicefrac{1}{\sqrt{2}},\beta ,0 \right)$ with known $\beta$. Here, the indirect inference is employed to the tvAR(1) in \eqref{eq:tvAR1} with the linear parametric form of the time varying coefficient $\alpha_1 \left(  u \right) =\theta_0+\theta_1 u$, and $\gamma\left( u \right)=\gamma_0+ \gamma_1 u$. The parameters of IM is $\theta=\left( \theta_0,\theta_1, \alpha ,\gamma_0 ,\gamma_1 \right)$. For AM, the same parametric form with the t-distribution assuming unknown $\nu$ is used, that is, $\lambda=(\theta_0^{(A)},\theta_1^{(A)},\nu,\gamma_0^{(A)},\gamma_1^{(A)})$. 

The simulation was performed by assuming $(\alpha,\beta,\theta_0,\theta_1,\gamma_0,\gamma_1)=(1.4,0, \allowbreak 0.35,-0.6,0.5,0.1)$. Table \ref{tab:simulation_tvAR1_alphaunkown} reports the MC mean and standard error of the estimates. Notice that the MC mean from the indirect estimates seems to be consistent. Table \ref{tab:simulation_tvAR1_alphaunkown_skew_kur} presents the kurtosis and skewness of indirect estimates. All indirect estimates do not present kurtosis close to 3 and the skewness close to 0. Indeed, they are similar to the case when $\alpha$ is known.

\begin{table}[!htbp]
	
	\centering
		\caption{MC mean and standard error for different sample size $T$ using indirect estimators assuming $(\alpha,\beta,\theta_0,\theta_1,\gamma_0,\gamma_1)=(1.4,0,0.35,-0.6,0.5,0.1)$ with known $\beta$ from $\alpha$-stable tvAR(1) based on $R=1000$ replications.}
	\label{tab:simulation_tvAR1_alphaunkown}
	\scriptsize
	\begin{tabular}{c ccccc|ccccc}
		\toprule
		\multirow{2}{*}{T} & \multicolumn{10}{c}{Indirect estimates} \\ \cline{2-11}
		& \multicolumn{5}{c}{Model of Interest} & \multicolumn{5}{c}{Auxiliary model}  \\ \hline
		& $\theta_0$ & $\theta_1$ & $\alpha$ & $\gamma_0$ & $\gamma_1$ & $\theta_0^{(A)}$ & $\theta_1^{(A)}$ & $\nu$ & $\gamma_0^{(A)}$ & $\gamma_1^{(A)}$ \\ 
		\midrule
		\multirow{2}{*}{$500$} & 0.3482 & -0.5980 & 1.4083 & 0.4922 & 0.1111 & 0.3482 & -0.5980 & 1.8853 & 0.3994 & 0.0897 \\ 
		& (0.0406) & (0.0715) & (0.0737) & (0.0527) & (0.0960) & (0.0407) & (0.0716) & (0.2351) & (0.0446) & (0.0778) \\ 
		\multirow{2}{*}{$1000$} & 0.3492 & -0.5986 & 1.4037 & 0.4974 & 0.1033 & 0.3492 & -0.5986 & 1.8622 & 0.4033 & 0.0834 \\ 
		& (0.0244) & (0.0430) & (0.0520) & (0.0370) & (0.0661) & (0.0244) & (0.0429) & (0.1570) & (0.0311) & (0.0533) \\ 
		\multirow{2}{*}{$1500$} & 0.3498 & -0.5988 & 1.4000 & 0.4976 & 0.1011 & 0.3499 & -0.5988 & 1.8478 & 0.4030 & 0.0818 \\ 
		& (0.0187) & (0.0323) & (0.0417) & (0.0305) & (0.0546) & (0.0187) & (0.0323) & (0.1244) & (0.0255) & (0.0441) \\ 
		\bottomrule
	\end{tabular}
\end{table}

\begin{table}[!htbp]
	\centering
		\caption{Kurtosis and skewness of indirect estimates and BWE for different sample size $T$ assuming ($\alpha,\beta,\theta_0,\theta_1,\gamma_0,\gamma_1$)=($1.4,0,0.35,-0.6,0.5,0.1$) with known $\beta$ from $\alpha$-stable tvAR(1) based on $R=1000$ replications.}
	\label{tab:simulation_tvAR1_alphaunkown_skew_kur}
	\begin{tabular}{c c|ccccc}
		\toprule
		\multirow{2}{*}{T} &  & \multicolumn{5}{c}{Indirect estimates} \\ \cline{3-7}
		& & $\theta_0$ & $\theta_1$ & $\alpha$ & $\gamma_0$ & $\gamma_1$  \\ 
		\midrule
		\multirow{2}{*}{$500$} & kur & 3.7767 & 3.6800 & 3.1078 & 2.8924 & 3.0343 \\ 
		& skw & -0.1369 & 0.1213 & 0.2730 & 0.1583 & 0.0156 \\ 
		\multirow{2}{*}{$1000$} & kur & 4.7209 & 3.8513 & 2.7680 & 3.1397 & 3.0710 \\ 
		& skw & -0.1548 & 0.1008 & 0.0889 & 0.0672 & -0.0654 \\ 
		\multirow{2}{*}{$1500$} & kur & 4.2029 & 3.8664 & 2.7385 & 3.0881 & 3.0266 \\ 
		& skw & 0.1274 & -0.0192 & 0.0967 & 0.0973 & 0.0108 \\ 
		\bottomrule
	\end{tabular}
\end{table}

Finally, Figure \ref{fig:tvAR_1_unknown_alpha1.4_density} shows the density estimates of each parameter. The density estimates show that the standard error become smaller as $T$ increases. We conclude that the distribution of indirect estimates seem to be consistent for these sample path length.

\begin{figure}[!htbp]
	\centering
	\includegraphics[width=\textwidth]{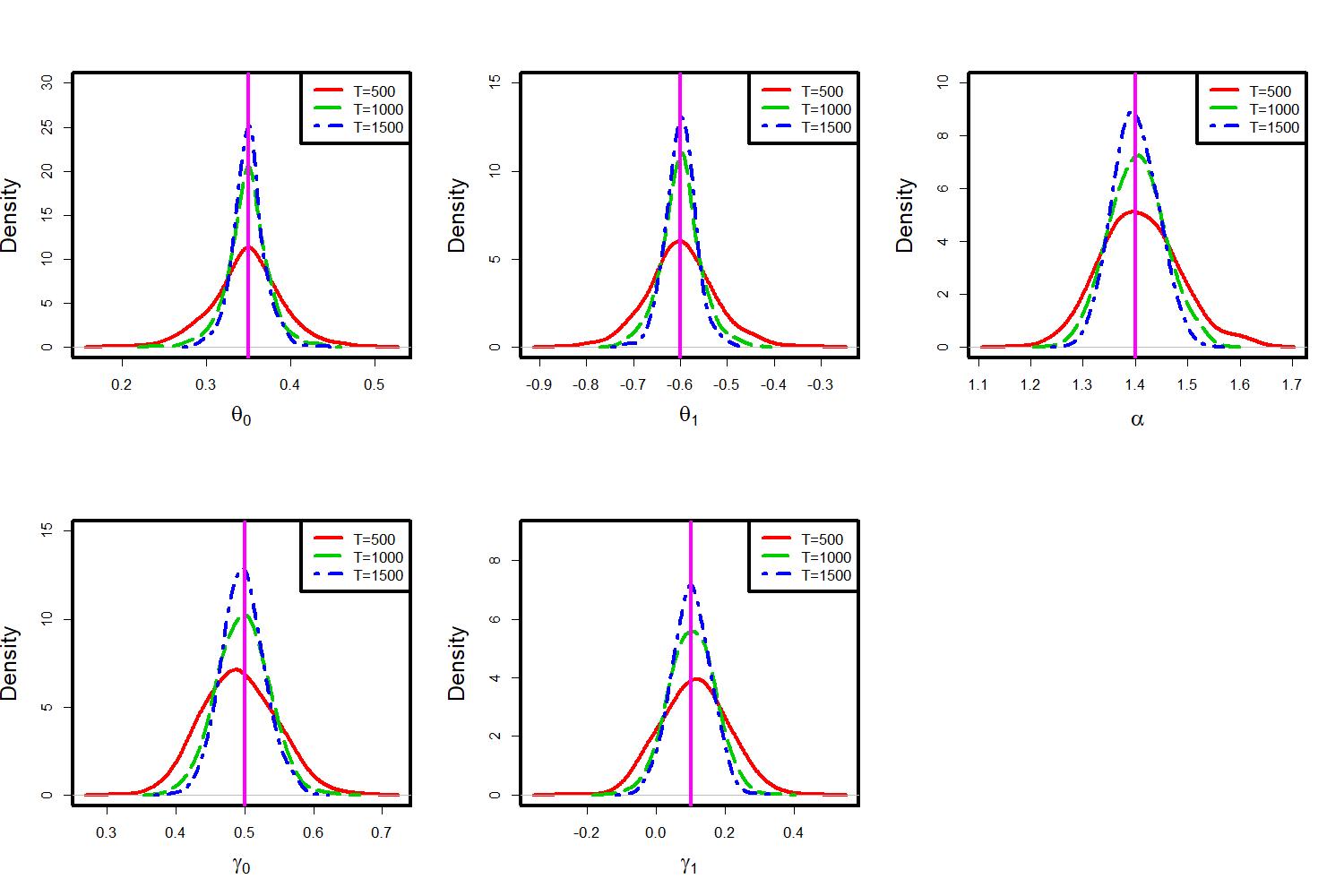}
	\caption{Density estimates of $\theta_0$, $\theta_1$, $\alpha$, $\gamma_0$ and $\gamma_1$ for different sample sizes based on $R=1000$ replications from $\alpha$-stable tvAR(1) with ($\alpha,\beta,\theta_0,\theta_1,\gamma_0,\gamma_1$)=($1.4,0,0.35,-0.6,0.5,0.1$) using indirect inference.} 
	\label{fig:tvAR_1_unknown_alpha1.4_density}
\end{figure}

\subsubsection{$\alpha$-stable tvMA(1)}

The indirect inference for the model \eqref{eq:tvMA1} with unknown $\alpha$ is illustrated. The parameter of IM is $\theta=(\theta_0,\theta_1,\alpha,\gamma)$ and the parameter of AM is $\lambda=(\theta_0^{(A)},\theta_1^{(A)},\nu,\gamma^{(A)})$. The simulation was performed by assuming $(\alpha,\beta,\theta_0,\theta_1,\gamma)=(1.75,0.2,-0.35, 0.4,0.7)$. 

The MC mean and standard error of the estimates from both model (IM and AM) are reported in Table \ref{tab:simulation_tvMA1_alphaunkown}, and kurtosis and skewness are presented in Table \ref{tab:simulation_tvMA1_alphaunkown_skew_kur}. Along with the density estimates showed in Figures \ref{fig:tvMA_1_unknown_alpha1.75_density}, the indirect estimates seem to be consistent with these sample path length. One interesting result is that while $\alpha<2$ implies the IM has infinite variance, the AM was estimated with $\nu>2$, i.e. finite variance.

\begin{table}[!htbp]
	\centering
	\caption{MC mean and standard error for different sample size $T$ using indirect estimators assuming $(\alpha,\beta,\theta_0,\theta_1,\gamma)=(1.75,0.2,-0.35,0.4,0.7)$ with known $\beta$ from $\alpha$-stable tvMA(1) based on $R=1000$ replications.}
	\label{tab:simulation_tvMA1_alphaunkown}
	\begin{tabular}{c cccc|cccc}
		\toprule
		\multirow{2}{*}{T} & \multicolumn{8}{c}{Indirect estimates} \\ \cline{2-9}
		& \multicolumn{4}{c}{Model of Interest} & \multicolumn{4}{c}{Auxiliary model}  \\ \hline
		& $\theta_0$ & $\theta_1$ & $\alpha$ & $\gamma$ & $\theta_0^{(A)}$ & $\theta_1^{(A)}$ & $\nu$ & $\gamma^{(A)}$ \\ 
		\midrule		
		\multirow{2}{*}{$500$}  & -0.3518 & 0.4016 & 1.7566 & 0.7008 & -0.3518 & 0.4016 & 3.9795 & 0.3810 \\ 
		& (0.0699) & (0.1245) & (0.0739) & (0.0296) & (0.0694) & (0.1237) & (1.0183) & (0.0390) \\ 
		\multirow{2}{*}{$1000$}  & -0.3487 & 0.3987 & 1.7527 & 0.6999 & -0.3486 & 0.3987 & 3.8307 & 0.3776 \\ 
		& (0.0446) & (0.0787) & (0.0559) & (0.0229) & (0.0445) & (0.0788) & (0.6414) & (0.0299) \\ 
		\multirow{2}{*}{$1500$}  & -0.3504 & 0.4009 & 1.7525 & 0.7003 & -0.3502 & 0.4007 & 3.7874 & 0.3785 \\ 
		& (0.0375) & (0.0663) & (0.0457) & (0.0187) & (0.0373) & (0.0661) & (0.4852) & (0.0242) \\ 		
		\bottomrule
	\end{tabular}
\end{table}

\begin{table}[!htbp]
	\centering
	\caption{Kurtosis and skewness of indirect estimates and BWE for different sample size $T$ assuming $(\alpha,\beta,\theta_0,\theta_1,\gamma)=(1.75,0.2,-0.35,0.4,0.7)$ with known $\beta$ from $\alpha$-stable tvMA(1) based on $R=1000$ replications.}
	\label{tab:simulation_tvMA1_alphaunkown_skew_kur}
	\begin{tabular}{c c|cccc}
		\toprule
		\multirow{2}{*}{T} &  & \multicolumn{4}{c}{Indirect estimates} \\ \cline{3-6}
		& & $\theta_0$ & $\theta_1$ & $\alpha$ & $\gamma_0$  \\ 
		\midrule
		\multirow{2}{*}{$500$} & kur & 3.8667 & 3.2718 & 2.8731 & 3.3445 \\ 
		& skw & -0.0413 & -0.0030 & -0.1140 & 0.0003 \\ 
		\multirow{2}{*}{$1000$} & kur & 3.7260 & 3.4565 & 2.8049 & 2.9412 \\ 
		& skw & 0.0436 & 0.0582 & -0.0081 & 0.1446 \\ 
		\multirow{2}{*}{$1500$} & kur & 3.6876 & 3.4211 & 3.0489 & 3.0002 \\ 
		& skw & -0.0043 & 0.0187 & -0.2133 & 0.0557 \\ 
		\bottomrule
	\end{tabular}
\end{table}

\begin{figure}[!htbp]
	\centering
	\includegraphics[width=\textwidth]{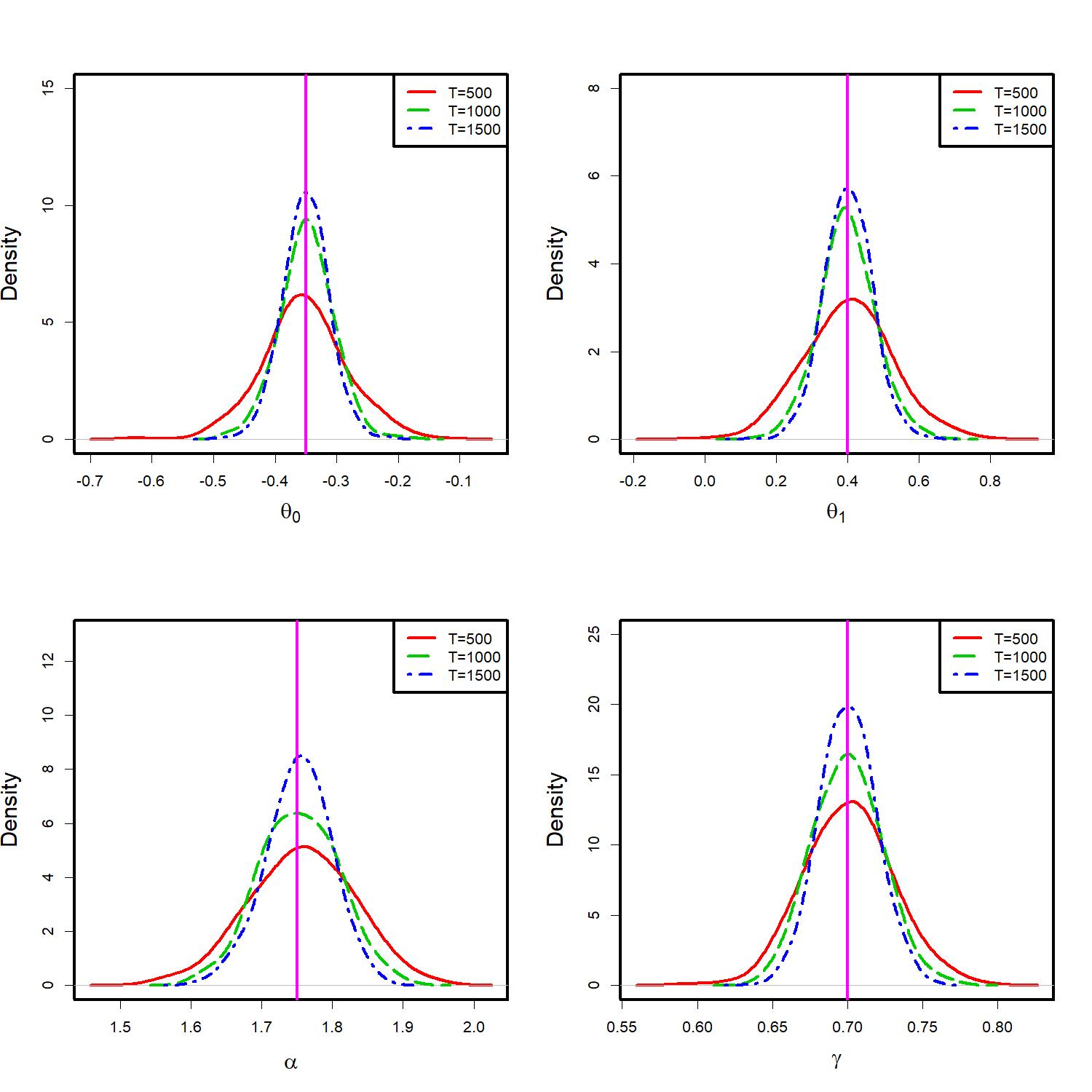}
	
	\caption{Density estimates of $\theta_0$, $\theta_1$, $\alpha$ and $\gamma$ for different sample sizes based on $R=1000$ replications from $\alpha$-stable tvMA(1) with $(\alpha,\beta,\theta_0,\theta_1,\gamma)=(1.75,0.2,-0.35,0.4,0.7)$ using indirect inference.} 
	\label{fig:tvMA_1_unknown_alpha1.75_density}
\end{figure}

\subsubsection{$\alpha$-stable tvARMA(1,1)}

Finally, the simulation was done for the case of tvARMA(1,1) in \eqref{eq:tvARMA1}, but $\alpha$ is assumed to be unknown. The time varying coefficients are assumed to be linear, i.e. $\alpha_1 \left(  u \right) =\theta_{a0}+\theta_{a1} u$ and $\beta_1 \left(  u \right) =\theta_{b0}+\theta_{b1} u$, and $\varepsilon_{t} \sim S_\alpha(\nicefrac{1}{\sqrt{2}},\beta,0)$ for known $\beta$. Therefore, the parameters of IM is $\theta=\left( \theta_{a0},\theta_{a1}, \theta_{b0},\theta_{b1}, \alpha, \gamma \right)$, while AM has the parameter $\lambda=\left( \theta_{a0}^{(A)},\theta_{a1}^{(A)}, \theta_{b0}^{(A)},\theta_{b1}^{(A)}, \nu, \gamma^{(A)} \right)$. The simulation was performed by assuming $(\alpha,\beta,\theta_{a0},\theta_{a1},\theta_{b0},\theta_{b1},\alpha,\gamma)=(1.3,0,-0.2,-0.4,0.2,0.3,1.1)$.

The MC mean and standard error of the estimates from both IM and AM are reported in Table \ref{tab:simulation_tvARMA1_alphaunkown}, and kurtosis and skewness are presented in Table \ref{tab:simulation_tvARMA1_alphaunkown_skew_kur}. The density estimates are showed in Figures \ref{fig:tvMA_1_unknown_alpha1.75_density}. Again, the indirect estimates seem to be consistent. Moreover, if we compare with simulation results from the known $\alpha$, they present similar standard error, kurtosis and asymmetry.

\begin{table}[!htbp]
	\centering
		\caption{MC mean and standard error for different sample size $T$ using indirect estimators assuming $(\alpha,\beta,\theta_{a0},\theta_{a1},\theta_{b0},\theta_{b1},\alpha,\gamma)=(1.3,0,-0.2,-0.4,0.2,0.3,1.1)$ with known $\beta$ from $\alpha$-stable tvARMA(1,1) based on $R=1000$ replications.}
	\label{tab:simulation_tvARMA1_alphaunkown}
	\begin{tabular}{c c cccccc}
		\toprule
		& T & $\theta_{a0}$ & $\theta_{a1}$ & $\theta_{b0}$ & $\theta_{b1}$ & $\alpha$ & $\gamma$ \\ \cline{2-8}
		\multirow{6}{*}{Model of Interest} & \multirow{2}{*}{$500$} & -0.2036 & -0.3932 & 0.1971 & 0.3064 & 1.3018 & 1.0923 \\ 
		&& (0.0585) & (0.0869) & (0.0587) & (0.0891) & (0.0698) & (0.0587) \\ 
		&\multirow{2}{*}{$1000$} & -0.2005 & -0.3986 & 0.2003 & 0.3004 & 1.3045 & 1.0976 \\ 
		&& (0.0319) & (0.0489) & (0.0329) & (0.0504) & (0.0471) & (0.0433) \\ 
		&\multirow{2}{*}{$1500$} & -0.1998 & -0.3999 & 0.2012 & 0.2983 & 1.2998 & 1.0953 \\ 
		&& (0.0233) & (0.0359) & (0.0250) & (0.0374) & (0.0390) & (0.0347) \\ 
		\midrule
		& T & $\theta_{a0}^{(A)}$ & $\theta_{a1}^{(A)}$ & $\theta_{b0}^{(A)}$ & $\theta_{b1}^{(A)}$ & $\nu$ & $\gamma^{(A)}$ \\ \cline{2-8}
		\multirow{6}{*}{Auxiliary model}  & \multirow{2}{*}{$500$} & -0.2036 & -0.3936 & 0.1971 & 0.3062 & 1.5904 & 0.7465 \\
		& & (0.0584) & (0.0864) & (0.0586) & (0.0889) & (0.1731) &( 0.0940) \\
		& \multirow{2}{*}{$1000$} & -0.2006 & -0.3986 & 0.2003 & 0.3006 & 1.5917 & 0.7542  \\
		& & (0.0319) & (0.0483) & (0.0329) & (0.0505) & (0.1160) & (0.0668) \\
		& \multirow{2}{*}{$1500$} & -0.1998 & -0.4009 & 0.2012 & 0.2983 & 1.5772 & 0.7487 \\
		& & (0.0232) & (0.0344) & (0.0250) & (0.0374) & (0.0947) & (0.0540) \\
		\bottomrule
	\end{tabular}
\end{table}

\begin{table}[!htbp]
	\centering
		\caption{Kurtosis and skewness of indirect estimates and BWE for different sample sizes ($T=500,1000,1500$) assuming $(\alpha,\beta,\theta_{a0},\theta_{a1},\theta_{b0},\theta_{b1},\alpha,\gamma)=(1.3,0,-0.2,-0.4,0.2,0.3,1.1)$ with known $\beta$ from $\alpha$-stable tvARMA(1,1) based on $R=1000$ replications.}
	\label{tab:simulation_tvARMA1_alphaunkown_skew_kur}
	\begin{tabular}{c c|cccccc}
		\toprule
		\multirow{2}{*}{T} &  & \multicolumn{5}{c}{Indirect estimates} \\ \cline{3-8}
		& & $\theta_{a0}$ & $\theta_{a1}$ & $\theta_{b0}$ & $\theta_{b1}$ & $\alpha$ & $\gamma$ \\ 
		\midrule
		\multirow{2}{*}{$500$} & kur & 5.2893 & 4.5345 & 6.0807 & 5.5860 & 3.0705 & 3.3815 \\ 
		& skw & 0.2593 & -0.1945 & -0.1951 & 0.1297 & 0.1406 & 0.2709 \\ 
		\multirow{2}{*}{$1000$} & kur & 4.6288 & 4.1073 & 4.5984 & 4.1306 & 3.4796 & 2.9077 \\ 
		& skw & -0.1406 & 0.1144 & 0.0360 & -0.0328 & 0.1167 & -0.0713 \\ 
		\multirow{2}{*}{$1500$} & kur & 4.9301 & 4.0790 & 5.1471 & 4.5091 & 3.1301 & 3.0701 \\ 
		& skw & 0.0236 & -0.1964 & 0.0964 & -0.2378 & 0.1004 & 0.0833 \\ 
		\bottomrule
	\end{tabular}
\end{table}

\begin{figure}[!htbp]
	\centering
	\includegraphics[width=\textwidth]{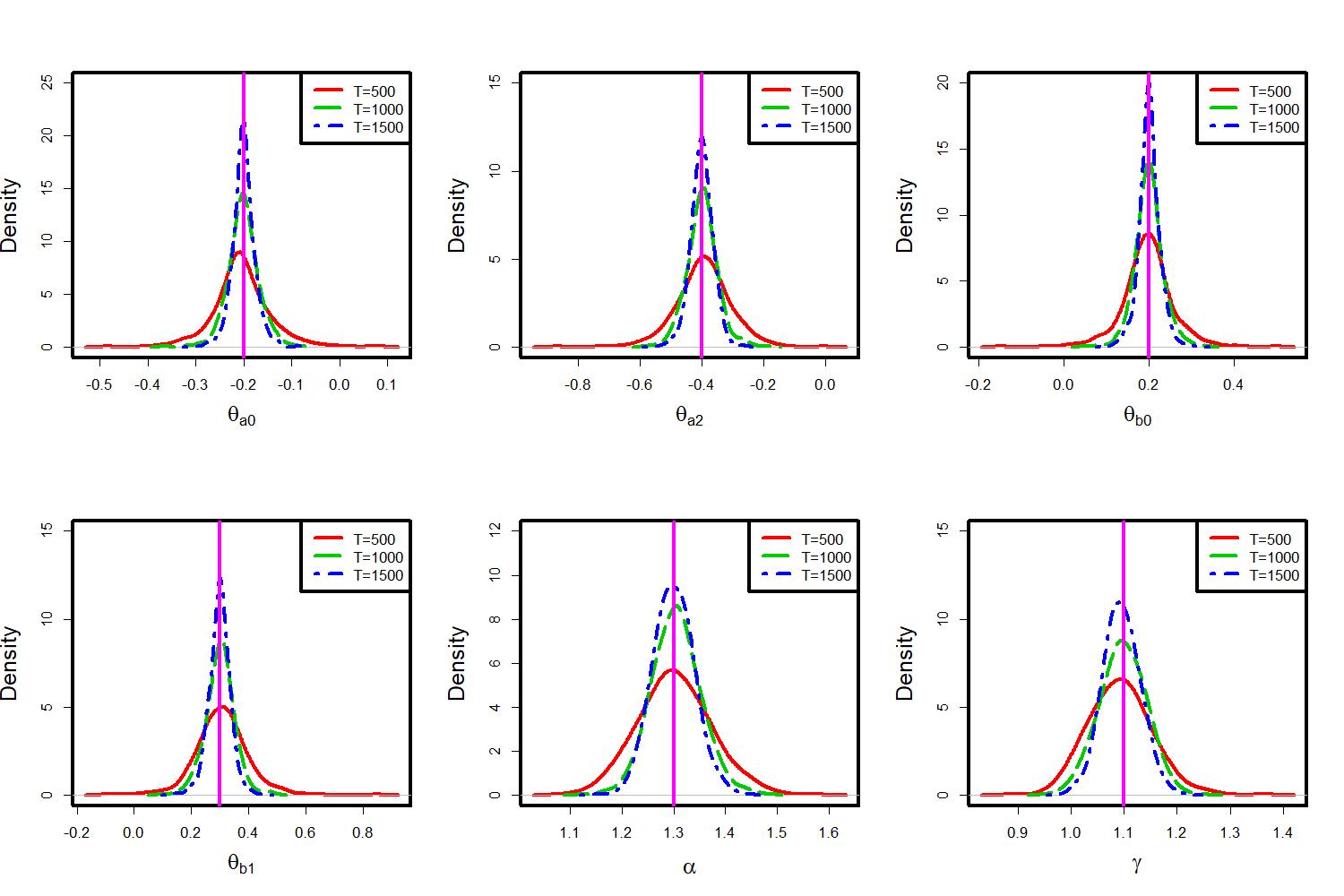}
	
	\caption{Density estimates of $\theta_{a0}$, $\theta_{a1}$, $\theta_{b0}$, $\theta_{b1}$, $\alpha$ and $\gamma$ for different sample sizes based on $R=1000$ replications from $\alpha$-stable tvARMA(1,1) with $(\alpha,\beta,\theta_{a0},\theta_{a1},\theta_{b0},\theta_{b1},\gamma)=(1.3,0,-0.2,-0.4,0.2,0.3,1.1)$ using indirect inference.} 
	\label{fig:tvARMA_1_unknown_alpha1.3_density}
\end{figure}

\section{Application}
\label{sec:application}

In this section, we illustrate an application for wind power generated in German offshore wind farms from 16/06/2015 at 00:00 to 27/07/2015 at 24:00 ($T=1008$ hours), obtained from the EMHIRES (European Meteorological High resolution RES time series) datasets \cite{Gonzales2016}. For daily data, the Gaussian innovation assumption seems to be appropriate, but the hourly time series present heavy tails and Gaussian assumption is inadequate. Figure \ref{fig:wind}, panel (a) shows the original time series ($y_t$) and its difference ($\Delta y_t$), while panel (b) shows the standardized histogram of the differenced time series, which shows heavy-tailed behavior. We select just a small segment of the data because the whole time series has more complex structure, such as seasonality, thus a non-parametric approach could be more appropriate.

\begin{figure}[!htbp]
	\centering
	\subfloat[Hourly wind power ($y_t$) and its difference ($\Delta y_t$)]{\includegraphics[width=.45\textwidth]{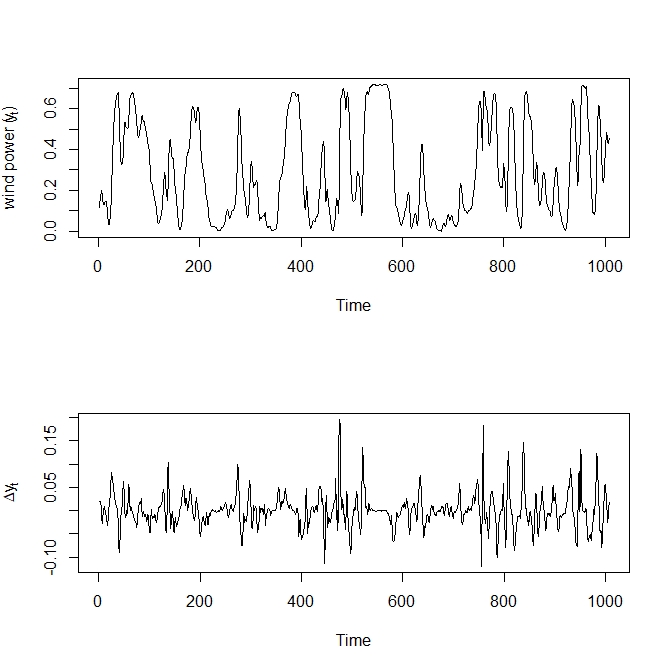}}
	\subfloat[Standardized histogram of $\Delta y_t$.]{\includegraphics[width=.45\textwidth]{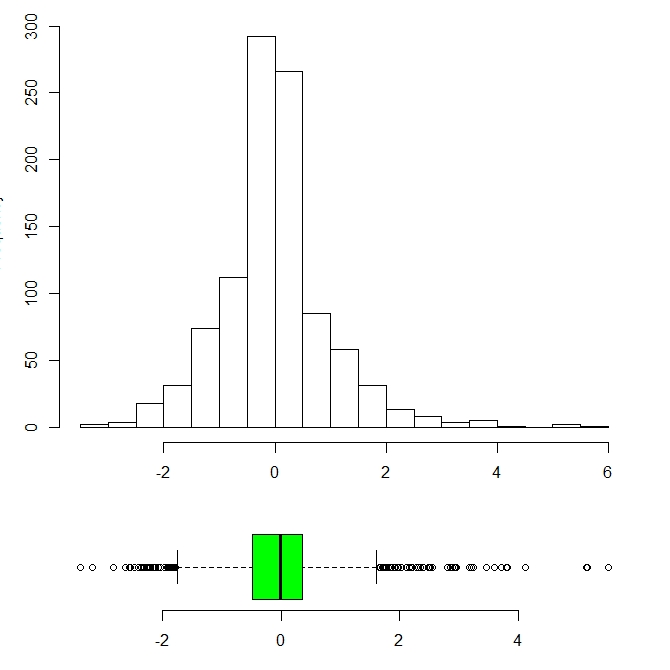}}
	\caption{Hourly wind power from 16/06/2015 at 00:00 to 27/07/2015 at 24:00.} 
	\label{fig:wind}
\end{figure}

Figure \ref{fig:wind_acf} shows sample autocorrelation function (global), and partial autocorrelation function. Traditional models, like ARMA(1,1) and AR(4) seem to be appropriate, but the blocked smooth periodogram shows its slowly changed structure over the time. 

\begin{figure}[!htbp]
	\centering
	\subfloat[ACF and partial ACF.]{\includegraphics[width=.45\textwidth]{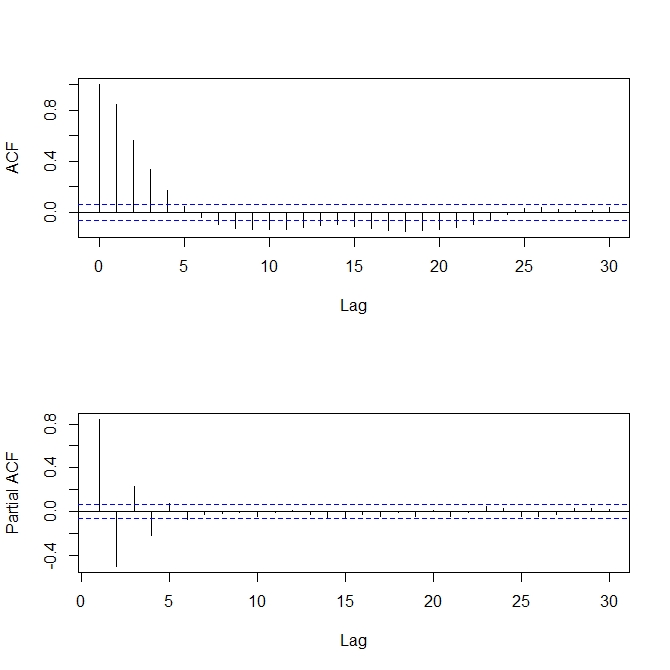}}
	\subfloat[Blocked smooth periodogram]{\includegraphics[width=.45\textwidth]{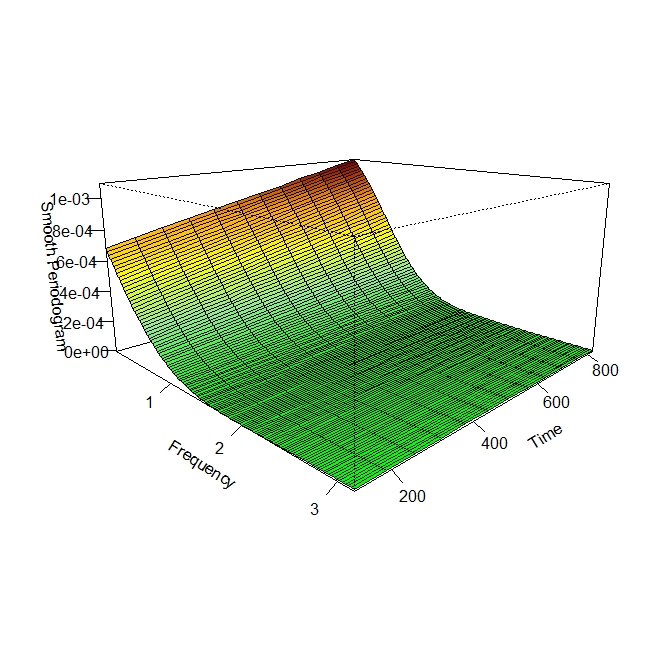}}
	\caption{ACF, partial ACF and blocked smooth periodogram of wind power data.} 
	\label{fig:wind_acf}
\end{figure}

To explore its local structure, we estimate ARMA(1,1) and AR(4) for 9 time blocks. Figures \ref{fig:wind_arma_local} and \ref{fig:wind_ar4_local} present the smoothed estimated coefficients over time for both models. Both cases show that coefficients are approximately linear over time. Consequently, two models are proposed: 
\begin{itemize}
	\item tvARMA(1,1) model with linear coefficients, $\alpha_1(u)=\theta_{a0}+\theta_{a1} (u)$, $\beta_1(u)=\theta_{b0}+\theta_{b1} (u)$ and $\gamma(u)=\gamma_0+\gamma_1 (u)$.
	\item tvAR(4) model with linear coefficients, $\alpha_1(u)=\theta_{a0}+\theta_{a1} (u)$, $\alpha_2(u)=\theta_{b0}+\theta_{b1} (u)$, $\alpha_3(u)=\theta_{c0}+\theta_{c1} (u)$, $\alpha_4(u)=\theta_{d0}+\theta_{d1} (u)$ and $\gamma(u)=\gamma_0+\gamma_1 (u)$.
\end{itemize}  

\begin{figure}[!htbp]
	\centering
	\includegraphics[width=\textwidth]{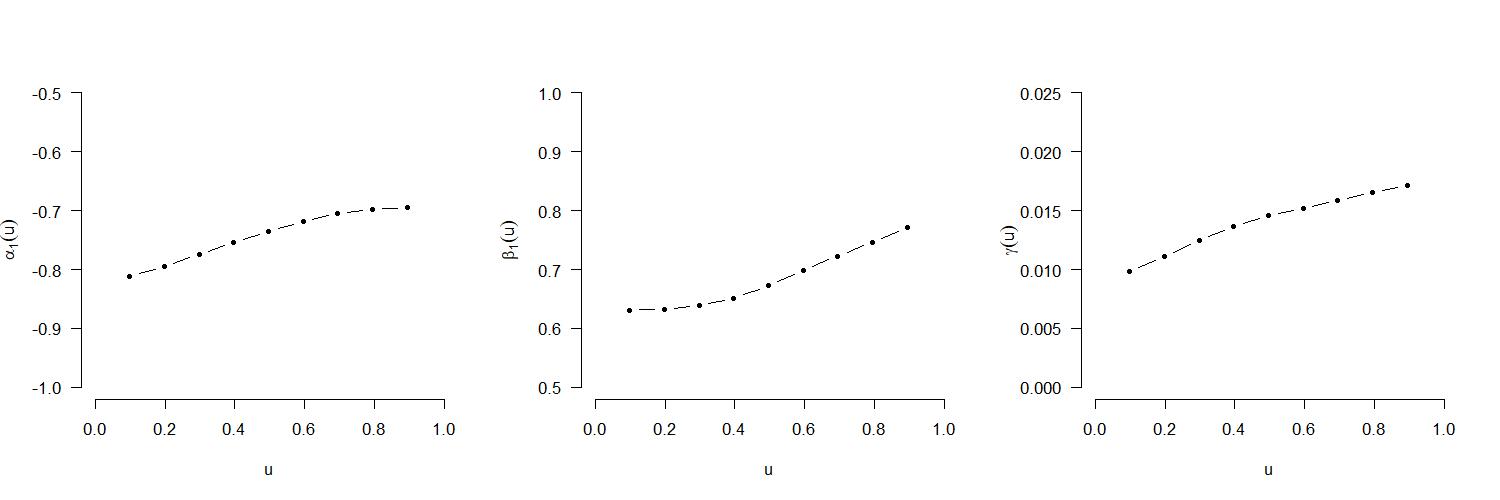}
	\caption{(Smoothed) $\alpha(u)$, $\beta(u)$ and $\gamma(u)$ estimates of stationary ARMA(1,1) model for $9$ block of size $N=200$ with $u=t/T$ center point of each block.} 
	\label{fig:wind_arma_local}
\end{figure}

\begin{figure}[!htbp]
	\centering
	\includegraphics[width=\textwidth]{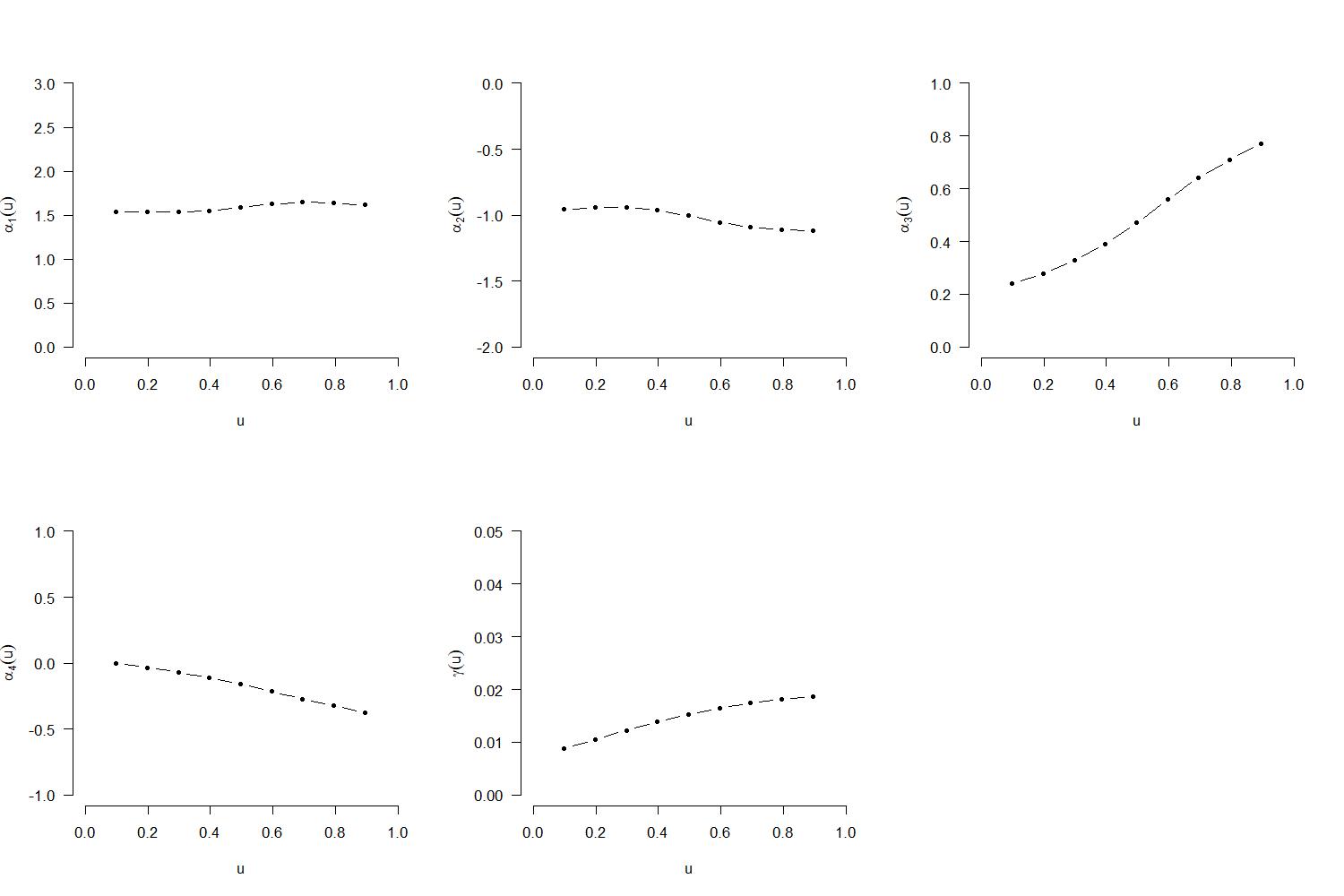}
	\caption{(Smoothed) $\alpha_1(u)$, $\alpha_2(u)$, $\alpha_3(u)$, $\alpha_4(u)$ and $\gamma(u)$ estimates of stationary AR(4) model for $9$ block of size $N=200$ with $u=t/T$ center point of each block.} 
	\label{fig:wind_ar4_local}
\end{figure}


After estimating both models, the residuals of tvARMA(1,1) are correlated, and we focus only on the tvAR(4). The parameter estimates are reported in Table \ref{tab:wind_blocked_whittle}. Figure \ref{fig:wind_ar4_residual} presents the residual analysis and the QQ-plot, box-plot and the histogram show that the distribution of error has heavy tail
and the residuals are approximately white noise. Additionally, we estimated the skewness ($0.35$) and kurtosis ($13.84$) and carried out Shapiro-Wilk and Jarque-Bera tests, which rejected the null hypothesis of normality. Moreover, Figure \ref{fig:wind_variogram} presents the variogram of the first difference of the wind data and the residuals from the tvAR(4) model. It is clear to observe that both of the variograms diverge.

\begin{table}[ht]
	\centering
		\caption{BWE of tvAR(4) from wind power time series.}
	\label{tab:wind_blocked_whittle}
	\begin{tabular}{c|cccc}
		\toprule
		\multirow{2}{*}{Parameter} & \multicolumn{4}{c}{BWE} \\
		\cline{2-5}
		& Estimate & s.e. & z-value & p-value \\ 
		\hline		
		$\theta_{a0}$ & -1.5985 & 0.0768 & -20.8171 & 0.0000 \\ 
		$\theta_{a1}$ & 0.3305 & 0.1406 & 2.3508 & 0.0187 \\ 
		$\theta_{b0}$ & 0.9135 & 0.1373 & 6.6536 & 0.0000 \\ 
		$\theta_{b1}$ & 0.0207 & 0.2447 & 0.0847 & 0.9325 \\ 
		$\theta_{c0}$ & -0.0585 & 0.1372 & -0.4266 & 0.6697 \\ 
		$\theta_{c1}$ & -0.7153 & 0.2445 & -2.9254 & 0.0034 \\ 
		$\theta_{d0}$ & -0.1316 & 0.0767 & -1.7158 & 0.0862 \\ 
		$\theta_{d1}$ & 0.5454 & 0.1405 & 3.8831 & 0.0001 \\ 
		$\gamma_0$ & 0.0077 & 0.0003 & 24.5452 & 0.0000 \\ 
		$\gamma_1$ & 0.0152 & 0.0007 & 21.6400 & 0.0000 \\ 
		\bottomrule
	\end{tabular}
\end{table}

\begin{figure}[!htbp]
	\centering
	\subfloat[QQ-plot]{\includegraphics[width=.4\textwidth]{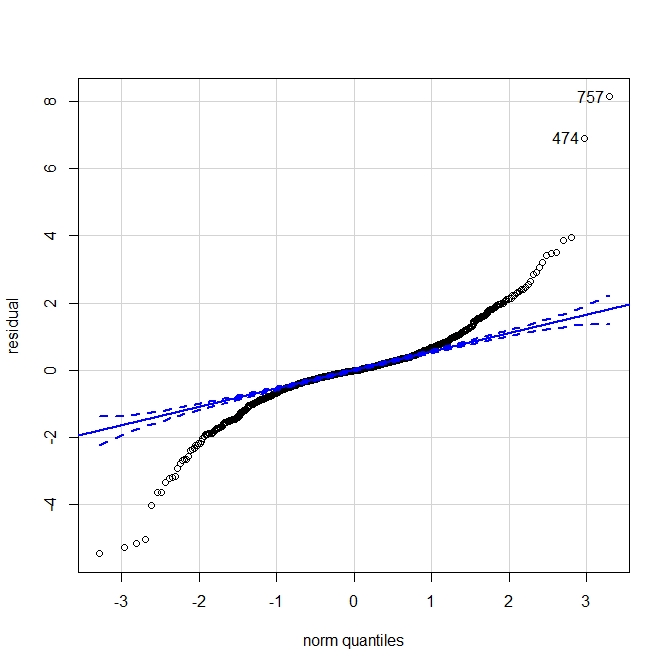}}
	\subfloat[Autocorreltion function]{\includegraphics[width=.4\textwidth]{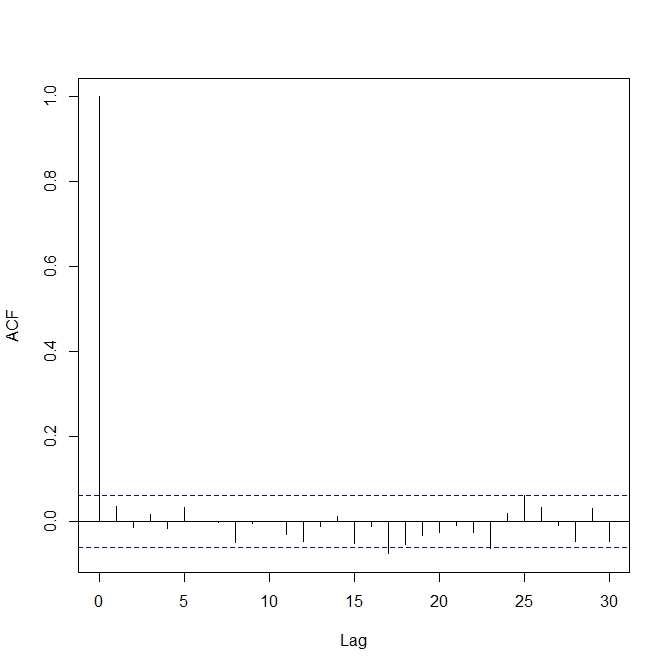}}
	
	\subfloat[Box-plot]{\includegraphics[width=.4\textwidth]{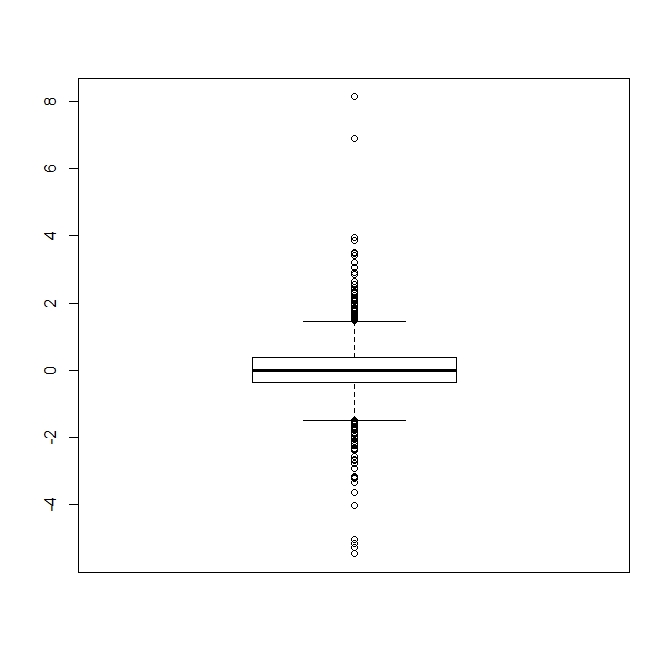}}
	\subfloat[Histogram with estimated stable curve ($\alpha=1.34$, $\beta=0$) and Gaussian curve.]{\includegraphics[width=.4\textwidth]{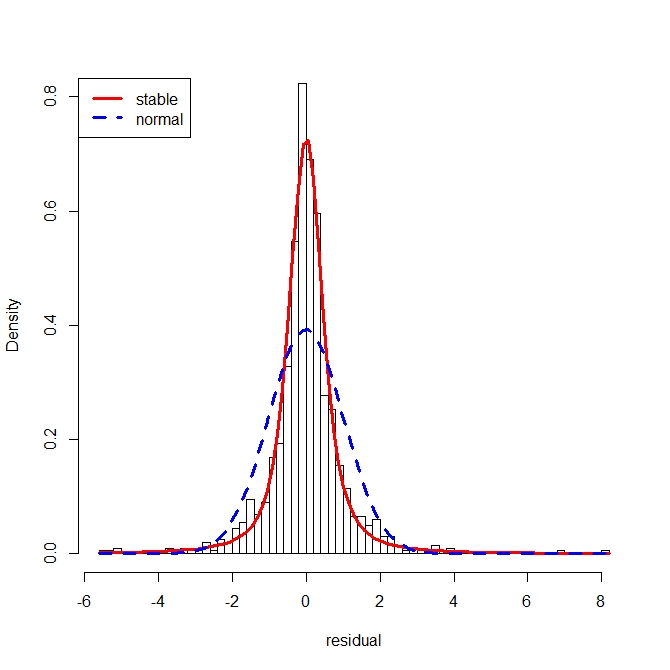}}
	\caption{Residual analysis using the BWE (standadized residual).} 
	\label{fig:wind_ar4_residual}
\end{figure}

\begin{figure}[!htbp]
	\centering
	\subfloat[$\Delta y_t$]{\includegraphics[width=.32\textwidth]{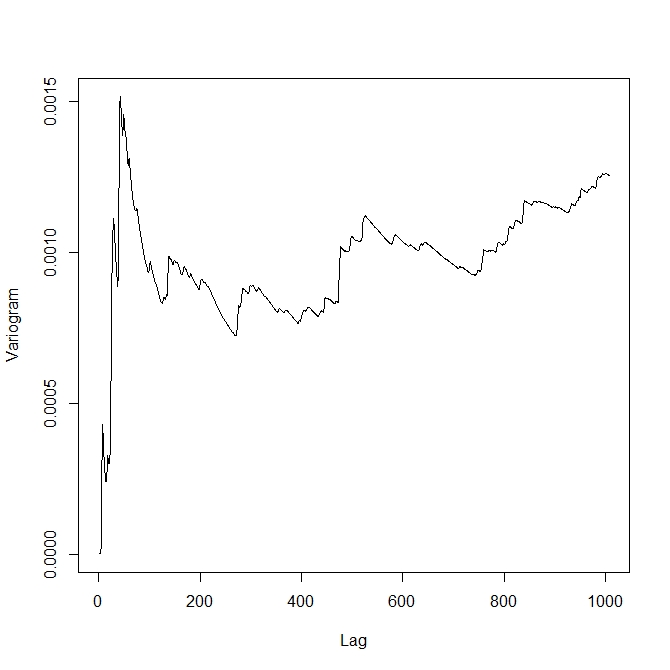}}
	\subfloat[Residuals]{\includegraphics[width=.32\textwidth]{fig13a.jpeg}}
	\caption{Variogram of the first differenced wind data $\Delta y_t$ and the residuals of the tvAR(4) model.} 
	\label{fig:wind_variogram}
\end{figure}

Since the residuals present heavy tail, we propose a more flexible model, $\alpha-$stable tvAR(4). We performed indirect estimation assuming known and unknown $\alpha$. In the first case, we assume $\alpha=1.34$ and $\beta=0$, which are obtained by MLE from the residuals of the BWE. Since estimation results are similar to the estimated model by assuming unknown $\alpha$, we present only results of the second model here.

The vector of parameters of IM is $(\theta_{a0},\theta_{a1},\theta_{b0},\theta_{b1},\theta_{c0},\theta_{c1},\theta_{d0},\theta_{d1}, \alpha, \gamma_{0},\gamma_{1})$ and the indirect inference was done by assuming symmetric $\alpha$-stable innovations. Table \ref{tab:wind_indirect_estimation_unknown_alpha} reports indirect estimates assuming $\beta=0$ with their MC standard error with $R=1000$ replications.

\begin{table}[!htbp]
	\centering
		\caption{Indirect estimates of $\alpha-$stable tvAR(4) with $S=40$ from wind data.}
	\label{tab:wind_indirect_estimation_unknown_alpha}
	\begin{tabular}{c c c}
		\toprule
		Parameter & Indirect estimate & Standard error \\ \hline
		$\theta_{a0}$ & -1.5434 & 0.0251 \\ 
		$\theta_{a1}$ & -0.0316 & 0.0426 \\ 
		$\theta_{b0}$ & 0.9036 & 0.0442 \\ 
		$\theta_{b1}$ & 0.1083 & 0.0764 \\ 
		$\theta_{c0}$ & -0.2818 & 0.0437 \\ 
		$\theta_{c1}$ & -0.2235 & 0.0752 \\ 
		$\theta_{d0}$ & 0.0639 & 0.0246 \\ 
		$\theta_{d1}$ & 0.1496 & 0.0412 \\
		$\alpha$ & 1.3875 & 0.0528 \\  
		$\gamma_{0}$ & 0.0065 & 0.0005 \\ 
		$\gamma_{1}$ & 0.0033 & 0.0010 \\ 
		\bottomrule
	\end{tabular}
\end{table}

To evaluate the residual distribution with the stable distribution, \cite{Nolan2002} suggested using the stabilized probability plot (stabilized p-p plot), proposed by \cite{Michael1983}, instead of the QQ-plot because the last one is not suitable to evaluate heavy-tailed distribution. In QQ-plot, large fluctuation for the extreme values in case of the heavy-tailed distribution produce large standard errors in the tails. Let $y_1 \leq \cdots \leq y_n$ be an ordered random sample of size $n$ from the distribution $F$. The stabilized p-p plot is defined as the plot of $s_i=\left( \frac{2}{\pi}\right)  \arcsin (F^{\frac{1}{2}}(y_i))$ versus $r_i=\left( \frac{2}{\pi}\right)  \arcsin \left( \left[ \left( i- \frac{1}{2} \right)/n  \right] ^{\frac{1}{2}}\right) $.  In this way, the histogram and the stabilized p-p plot in figure \ref{fig:wind_arma_ind_inf_residual} show that the stable distribution fits well the residuals.

\begin{figure}[!htbp]
	\centering
	\subfloat[Histogram with the stable curve]{\includegraphics[width=.4\textwidth]{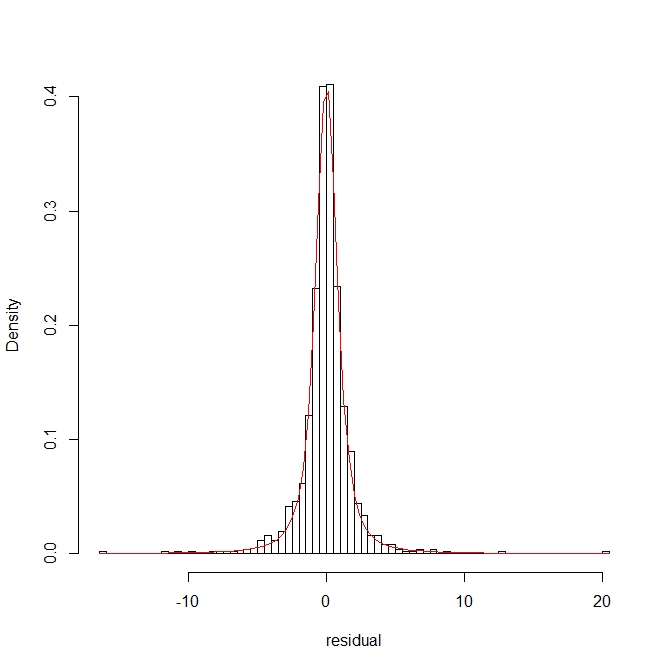}}
	\subfloat[Stabilized p-p plot]{\includegraphics[width=.4\textwidth]{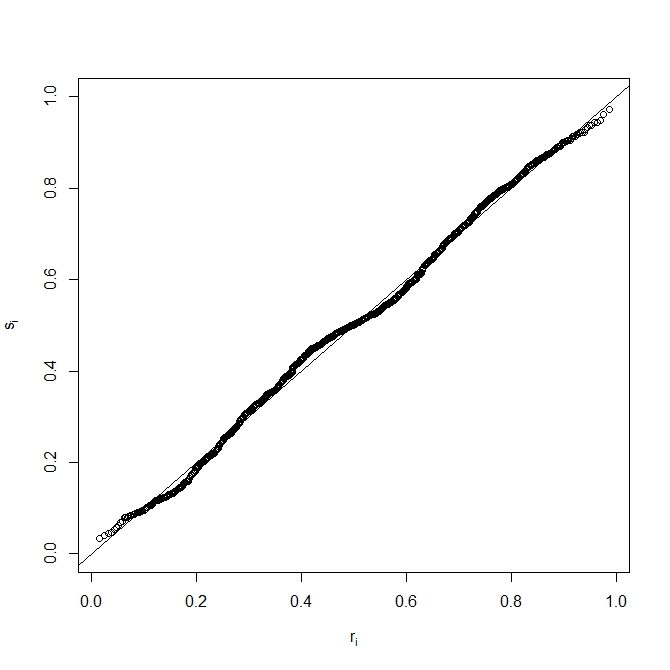}}\\
	\caption{Residual analysis from the $\alpha-$stable tvAR(4) model assuming that the innovation distribution is stable with $\alpha=1.34$ and $\beta=0$.} 
	\label{fig:wind_arma_ind_inf_residual}
\end{figure}

Finally, we compare the Mean square error (MSE), Root mean square error (RMSE) and Mean absolute error (MAE) of tvARMA(1,1) and tvAR(4) using BWE and indirect estimates. Note that MSE and RMSE do not make sense theoretically if we assume $\alpha-$stable tvAR(4). In Table \ref{tab:wind_indirect_error}, we observe that using BWE (assuming finite variance), MSE and RMSE are slightly lower, while using the indirect inference presents lower MAE.

Since the residual analysis indicates heavy-tails, $\alpha-$stable tvAR(4) is a better model to describe the data. In this case, by assuming $\alpha=1.34$, which is far from $2$, the simulation done in the previous section shows that the BWE is not appropriate. Even though MSE and RMSE are lower for BWE, they are not appropriate for $\alpha-$stable process since they cannot be theoretically handled. Based on MAE, the indirect inference performs slightly better. Moreover, the interpretation of estimated coefficients of the model also changed, i.e. $\alpha_1(u)$ and $\alpha_2(u)$ are constant, while $\alpha_3(u)$, $\alpha_4(u)$ and $\gamma(u)$ vary linearly. 

\begin{table}[!htbp]
	\centering
	\caption{Goodness of fit of different models for the wind data.}
	\label{tab:wind_indirect_error}
	\begin{tabular}{l c c c c}
		\toprule
		Model & MSE & RMSE & MAE \\ 
		\midrule
		tvARMA(1,1) & 0.000248 & 0.015739 & 0.009675 \\
		$\alpha$-stable tvARMA(1,1) & 0.000257 & 0.016028 & 0.009469  \\
		tvAR(4) & 0.000242 & 0.015542 & 0.009468 \\
		$\alpha$-stable tvAR(4) & 0.000256 & 0.015993 & 0.009094 \\
		\bottomrule
	\end{tabular}
\end{table}

\section{Conclusion}
\label{sec:conclusion}
In this paper, we studied $\alpha-$stable locally stationary ARMA processes and presented their properties. In contrast to the locally stationary processes with finite variance, this process involves the infinite variance observed in different fields. We also proposed an indirect inference method for the process with parametric time-varying coefficients. We performed simulations for basic models with linear parametric coefficients for known and unknown $\alpha$. The results show that indirect inference appropriate. An application is also illustrated.

There are some limitations that still need to be solved in the future. Firstly, since the time-varying spectral representation does not exist, identifying the local structure using traditional methods (autocorrelation and partial autocorrelation) are an informal way to identify the time-varying structure. One possibility is the local version of the dependence measure called autocovariation \citep{Kokoszka1994}. Secondly, simulations should be done for more complex models and also consider the possibility of non-parametric models. Thirdly, the indirect inference is time-consuming but they are appropriate when heavy-tailed innovations are present. Simulations suggest that when $\alpha$ is close to $2$, Gaussian innovations can be assumed. Model selection is still an open question. Also, there is few work related to prediction.

Finally, we are involved in research about the locally stationary process with tempered stable innovations, which are similar to stable distribution in its center but their tails are lighter and moments of all orders are finite.


\section*{Funding}

The authors are grateful to the support of a CNPq grant (141607/2017-3) and the University of Costa Rica (SWC), and a Fapesp grant 2018/04654-9 (PAM).

\bibliographystyle{tfnlm}

\end{document}